\documentclass[10pt,journal,twoside]{IEEEtran}
\usepackage{graphicx}
\usepackage{epstopdf}
\usepackage{placeins}
\usepackage{array}
\DeclareGraphicsExtensions{.eps,.pdf} 
\graphicspath{ {figures/} }
\usepackage{cite}
\usepackage{balance}
\usepackage{amssymb,amsmath}
\usepackage{subfigure}
\usepackage[ruled,vlined]{algorithm2e}
\usepackage{amssymb,amsmath,mathtools}
\usepackage{amssymb,amsmath}
\usepackage{algorithmic}
\usepackage{color}
\usepackage{multirow}
\usepackage{stackrel,amssymb,amsmath}
\usepackage{empheq}
\usepackage{cases}

\newcolumntype{P}[1]{>{\centering\arraybackslash}p{#1}}

\makeatletter
\newcommand\restartchapters{\par
  \setcounter{chapter}{0}%
  \setcounter{section}{0}%
  \gdef\@chapapp{\chaptername}%
  \gdef\thechapter{\@arabic\c@chapter}}
\makeatother

\newtheorem{theorem}{\bf Theorem}
\newtheorem{proof}{\bf Proof}

\newtheorem{remark}{Remark}

\makeatletter
\def\endthebibliography{%
	\def\@noitemerr{\@latex@warning{Empty `thebibliography' environment}}%
	\endlist
}
\makeatother

\begin{document}
\bstctlcite{IEEEexample:BSTcontrol}
\title{\huge Security-Reliability Trade-Off Analysis for SWIPT- and AF-Based IoT Networks with Friendly Jammers }
\author{\normalsize
\IEEEauthorblockN{$\text{Tan N. Nguyen}, \textit{Member, IEEE}$, $\text{Dinh-Hieu Tran}, \textit{Graduate Student Member, IEEE}$, $\text{Trinh Van Chien}, \textit{Member, IEEE}$, $\text{Van-Duc Phan}$, $\text{Miroslav Voznak}, \textit{Senior Member, IEEE}$, $\text{Phu Tran Tin}$, $\text{Symeon Chatzinotas}, \textit{Senior Member, IEEE},$ $\text{Derrick Wing Kwan Ng}, \textit{Fellow, IEEE}$, and $\text{H. Vincent Poor}, \textit{Life Fellow, IEEE}$  \vspace*{-0.8cm} }

\thanks{Manuscript received XXX; revised XXX;
	accepted XXX. Date of publication XXX; date of
	current version XXX. This research was financially supported by Van Lang University, Vietnam and is supported by Industrial University of Ho Chi Minh City (IUH) under grant number 131/HD-DHCN, in part by the Australian Research Council’s Discovery Project (DP210102169), and in part by the Ministry of Education, Youth and Sports of the Czech Republic under the grant SP2021/25 and e-INFRA CZ (ID:90140). (\textit{Corresponding author: Van-Duc Phan.})}
\thanks{Tan N. Nguyen is with the Communication and Signal Processing Research Group, Faculty of Electrical and Electronics Engineering, Ton Duc Thang University, Ho Chi Minh City, Vietnam. (e-mail:nguyennhattan@tdtu.edu.vn).}
\thanks{Dinh-Hieu Tran and Symeon Chatzinotas are with the Interdisciplinary Centre for Security, Reliability and Trust (SnT), the University of Luxembourg, Luxembourg. (e-mail: \{hieu.tran-dinh, symeon.chatzinotas\} @uni.lu).}
\thanks{Trinh Van Chien is with the School of Information and Communication Technology (SoICT), Hanoi University of Science and Technology, Vietnam. (e-mail: chientv@soict.hust.edu.vn.)}
\thanks{Derrick Wing Kwan Ng is with the School of Electrical Engineering and Telecommunications, University of New South Wales, Sydney, NSW 2025, Australia. (e-mail: w.k.ng@unsw.edu.au).}
\thanks{H. V. Poor is with the Department of Electrical and Computer Engineering, Princeton University, Princeton, NJ 08544 USA. (email: poor@princeton.edu).}
\thanks{Miroslav Voznak is with VSB - Technical University of Ostrava, 17. listopadu 15/2172, 708 33 Ostrava - Poruba, Czech Republic. (e-mail:miroslav.voznak@vsb.cz).}
\thanks{Van-Duc Phan is at Faculty of Automobile Technology, Van Lang University, Ho Chi Minh City, Vietnam. (email: duc.pv@vlu.edu.vn).}
\thanks{Phu Tran Tin is with Faculty of Electronics Technology, Industrial University of Ho Chi Minh City, Ho Chi Minh City, Vietnam (e-mail:phutrantin@iuh.edu.vn).}
}
\maketitle
\thispagestyle{empty}
\pagestyle{empty}
\begin{abstract}
Radio-frequency (RF) energy harvesting (EH) in wireless relaying networks has attracted considerable recent interest, especially for supplying energy to relay nodes in Internet-of-Things (IoT) systems to assist the information exchange between a source and a destination. Moreover, limited hardware, computational resources, and energy availability of IoT devices have raised various security challenges. To this end, physical layer security (PLS) has been proposed as an effective alternative to cryptographic methods for providing information security. In this study, we propose a PLS approach for simultaneous wireless information and power transfer (SWIPT)-based half-duplex (HD) amplify-and-forward (AF) relaying systems in the presence of an eavesdropper. Furthermore, we take into account both static power splitting relaying (SPSR) and dynamic power splitting relaying (DPSR) to thoroughly investigate the benefits of each one. To further enhance secure communication, we consider multiple friendly jammers to help prevent wiretapping attacks from the eavesdropper. More specifically, we provide a reliability and security analysis by deriving closed-form expressions of outage probability (OP) and intercept probability (IP), respectively, for both the SPSR and DPSR schemes. Then, simulations are also performed to validate our analysis and the effectiveness of the proposed schemes. Specifically, numerical results illustrate the non-trivial trade-off between reliability and security of the proposed system. In addition, we conclude from the simulation results that the proposed DPSR scheme outperforms the SPSR-based scheme in terms of OP and IP under the influences of different parameters on system performance.
\end{abstract}

\begin{IEEEkeywords}
Amplify-and-forward, dynamic power splitting, intercept probability, outage probability, source selection, SWIPT.
\end{IEEEkeywords}

\vspace*{-0.8cm}
\section{Introduction} \label{Introduction}
The Internet of Things (IoT) has played a pivotal role in fifth generation (5G) and beyond networks, seen as a novel solution to enable a smarter and safer life via autonomous monitoring and control in fields such as healthcare, manufacturing, and agriculture as demonstrated in  \cite{WangPLS2019,Ge5G2020,LiuToward5G2020,PhuBackMEC,CongIoT6G,FangIoT6G,Tran2020FDUAV} and the reference therein. Nevertheless, the enormous number of potential IoT devices also impose technical challenges in wireless communication due to limited resources comprising, for example, available bandwidth and energy supply. In particular, replacing or recharging batteries for IoT devices is generally costly, inconvenient, and even impossible in many scenarios such as hazardous or toxic environments or inside the human body. Smart ways of using and harvesting energy in IoT devices have attracted particular interest.
\subsection{Related Works}
To address these issues, energy harvesting (EH) communication networks are considered as alternative solutions. Energy can be harvested from sources such solar \cite{hieu2016stability}, wind \cite{JungInvestWind,Yenselfwind}, vibration \cite{SangVibration,QiuVibration}, and radio frequency (RF) signals, among which RF-based EH is an attractive solution because of its controllability and predictability. More importantly, it can carry both energy and information. Based on the above discussion, simultaneous wireless information and power transfer (SWIPT) is a promising future direction for realizing sustainable wireless communication \cite{tran2020throughput}. There are two kinds of EH receivers used in SWIPT networks: time switching (TS) and power splitting (PS) techniques. For the TS technique, the receiving node switches between information transfer (IT) and EH in different time slots, whereas, in the PS method, it splits the received power into factors for IT and EH \cite{12,hieu2018performance}.

Beyond the benefits of EH, relay nodes in cooperative communication networks can help a source node transfer information to a destination, which can extend the coverage of IoT devices with inherent limitations, such as low power and remote location \cite{PetrovVehicle2018,TungIoT2021,van2021controlling,tran2021satellite}. Therefore, cooperative relay networks with EH have received significant attention from researchers in recent years \cite{13,14,15,16,17,18,tin2020exploiting,nguyen2020wireless}. For example, amplify-and-forward (AF)-based wireless cooperative or sensor networks with TS and PS protocols were considered \cite{13}. Also, Chen et al. \cite{14} proposed and investigated a novel multi-hop AF relaying network in terms of co-channel interference (CCI) and Nakagami-$m$ fading, whereas each user could harvest energy from the CCI. In contrast to \cite{13, 14, van2020coverage} which considered only SISO systems, a multiple-input multiple-output (MIMO) system for maximizing the efficiency of SWIPT was investigated in \cite{15}. Furthermore, the system performance of cognitive radio networks (CRNs) was studied in \cite{16,17}, and the system performance of bidirectional relay networks was considered in \cite{18}. Besides, Tin et al. \cite{tin2020exploiting} proposed a new EH-based two-way (TW) half-duplex (HD) relay sensor network in the presence of a direct link between a transmitter and a receiver. Specifically, they derived the exact and asymptotic ergodic capacity and performed an exact analysis of symbol error rate. Furthermore, the authors in \cite{nguyen2020wireless} also proposed and investigated a new system model for SWIPT-based TW relaying systems. Therein, they derived closed-form expressions for the outage probability (OP) of three relaying schemes, termed decode-and-forward (DF), AF, and hybrid-decode-and-forward (HDAF). In addition, An et al. \cite{hoang2020physical} considered hybrid time–power splitting (HTPS) TW HD cooperative relaying in the presence of an eavesdropper. In this context, they derived closed-form expressions for OP and IP using the maximal ratio combining (MRC) and the selection combining (SC). The authors in \cite{tin2020power} investigated physical layer security (PLS) in a power beacon-assisted full-duplex (FD) EH relaying system using delay-tolerant (DT) and delay-limited (DL) methods. It was shown that the OP can be significantly improved by applying a dynamic PS scheme in \cite{42} and an adaptive PS scheme as in \cite{43,44}. The authors in \cite{42} considered a novel dynamic asymmetric PS scheme based on asymmetric instantaneous channel gains between relay and destinations to improve the OP. In \cite{43}, the authors proposed an adaptive PS protocol to enhance OP performance and average achievable capacity. Moreover, the authors in \cite{44} took into account the combination of the TS protocol and the dynamic PS to demonstrate the superiority of their proposed scheme in terms of OP and data transmission rate over the TS and PS schemes. 

Due to the broadcast nature of the wireless medium, information in IoT networks can be easily overheard, thus, the problem of enhancing security in IoT communications is an important issue. In comparison with conventional upper-layer security methods, PLS has many advantages such as 1) uncomplicated secret key distribution and management due to independence in encryption/decryption operations; 2) simple signal processing operations that involve minor additional overhead; and 3) adaptive signal design and resource allocation with flexible security levels. Based on these advantages,  PLS represents a promising solution for IoT networks \cite{19,20,22,23,24,25}. In \cite{27}, secrecy outage performance (SOP) using transmit antenna selection (TAS)/SC of multi-hop cognitive wireless sensor networks was investigated, and secure performance in a dual-hop MIMO relay system with outdated channel state information (CSI) was evaluated in \cite{26}. Also, in \cite{dinh2018secrecy}, the authors proposed and analyzed a generalized partial relay selection (PRS) protocol to improve security for CRNs with both cases of perfect or imperfect CSI. Moreover, in \cite{HieuIEEEsensor}, the authors also considered a novel system model for an EH-based PLS multi-hop multi-path cooperative wireless network. They then proposed three relay protocols, termed the shortest path, the random path, and the best path selection schemes to enhance PLS performance. A power allocation scheme to improve the PLS of the relay network was presented in \cite{33}, and the authors in \cite{34} proposed user selection along with an antenna selection (AS) scheme to maximize the end-to-end signal-to-noise ratios of a cellular multi-user two-way AF relay network. In addition, a novel wireless caching scheme to enhance the PLS of video streaming in cellular networks was proposed in \cite{35,36}; and the influence of an eavesdropper and CCI on intercept probability (IP) was considered in \cite{37}. 

Despite the fruitful research in the literature, the aforementioned works in \cite{19,20,22,23,24,25,27,26,dinh2018secrecy,HieuIEEEsensor,33,34,35,36,37} did not take into account jammers or artificial noises (ANs) to improve system security. Recently, intensive works have brought ANs and jammers into consideration \cite{4595041,6172252,5940246,7018202,7229350,6678044,6985747}. To prevent the eavesdropper from intercepting transmitted signals, Goel and Negi in \cite{4595041} studied the employment of multiple transmit antennas for generating artificial noise to interfere with the eavesdropper without disturbing the legitimate receiver. It was shown in \cite{4595041} that wireless secrecy can be guaranteed with the aid of AN if the transmitter has more antennas than the eavesdropper. In \cite{6172252}, the authors investigated secure communication of a MIMO system, in which the source, destination, and eavesdropper are each equipped with a random number of antennas. In \cite{5940246,7018202,7229350}, distributed beamforming or friendly jammers were deployed in the relay networks to prevent the eavesdropper from overhearing confidential information. In \cite{6678044,6985747}, the authors studied multiuser scheduling as a method to enhance the system PLS.

\subsection{Motivation and Contributions}
Despite previous achievements in the above-mentioned works, the investigation of SWIPT and PLS in cooperative IoT networks still has room for research in terms of the closed-form expressions for the OP and IP, which are independent of small-scale fading coefficients and can work for a long period of time. Moreover, to the best of our knowledge, there is no such related work on designing the optimal dynamic PS scheme for a SWIPT- and AF-based relay network consisting of multiple sources, multiple friendly jammers, and an EH relay, and a destination in the presence of an eavesdropper. Motivated by the above discussions, this paper provides a thorough analysis of the reliability and security trade-off using PS-based relay protocol. The best source approach is proposed to enhance performance with the effect of an eavesdropper and the presence of friendly jammers to prevent this eavesdropper. The main contributions of this research can be summarized as follows:
\begin{itemize}
	\item We derive closed-form expressions for the OP of the legitimate communications and an exact integral-form for the IP of the eavesdropper's channel under the assistance from multiple friendly jammers. 
	\item This work also provides an in-depth analysis of the influence of various system parameters on security and reliability performance. For the SWIPT technique, both dynamic power splitting-based relaying (DPSR) and static power splitting-based relaying (SPSR) are considered in our work to give a full picture of the advantages of each method for OP and IP cases. 
	\item Mathematical analysis is given to obtain an exact closed-form expression for the optimal power splitting ratio, i.e., $\rho^\star$, to maximize the total achievable rate at the destination.
	\item The correctness of our analysis is confirmed through numerical simulations. From the simulation results, we provide recommendations on selecting the configurations to obtain reliable and secure transmission without paying too much for the complexity of the system.
\end{itemize}

The rest of the paper is organized as follows. The system model are given in Section~\ref{System_Model}. The performance is presented in Section~\ref{sec:3}. Numerical results are depicted in Section~\ref{Sec:Num}, and Section~\ref{Sec:Con} concludes the paper.


\section{System Model}
 \label{System_Model}
As described in Fig.1, the system model includes multi-source (M source nodes) transfer information to the destination (D) via one relay (R) with the presence of multi friendly jammers (${{\rm{J}}_1},...,{{\rm{J}}_{\rm{K}}}$) and an eavesdropper (E). Friendly jammers are legitimate users in the system and emit artificial noise using a pseudo-random sequence, which is known to other legitimate users (i.e., source, relay, and destination) and remains unknown to eavesdropper E \cite{YulongTWC}. It should be noted that these pseudo-random sequences are known among legitimate users but not to the illegitimate user. Moreover, similar to conventional cryptography, pseudo-random sequences do not have to be pre-shared, they can be acquired by legitimate users through physical layer key agreement and generation that supports channel estimation, as intensively studied in the information-theoretic PLS literature \cite{ye2010information,bloch2011physical}. Particularly, the wireless system is assumed to change the pseudo-random sequence frequently, which significantly reduces the possibility that an eavesdropper can know all of these sequences and enhances system security. 

We assume that the direct connection between the $M$ sources and the destination nodes is too weak due to severe fading or long distances, hence, the only available communication path as well as power transfer path is through relay R. We also assume that E is located far away from M sources and cannot overhear the messages transmitted from these sources. All nodes in this model are single-antenna devices and operate on half-duplex (HD) mode. The energy harvesting (EH) and information transmission processes for the model system are shown in Fig.~\ref{fig:2}. In Fig.~\ref{fig:2}, we exploit the first phase to supply the required power to the relay from the best source S to help the relay forward the exchange data afterward. Specifically, the relay R utilizes a power splitter to divide the received signal into two parts, whereas the first part is used for collecting energy and storing it in the battery, while the second part is used for transferring received data. We define $\rho$ as a power splitting ratio, i.e., the ratio of the power received at the relay node used for the EH. Next, the IT process from R to D is conducted in the remaining time period \cite{46,47,48}. 

\begin{figure*} [t]
	\begin{center}
		\includegraphics[width=14cm,height=8cm]{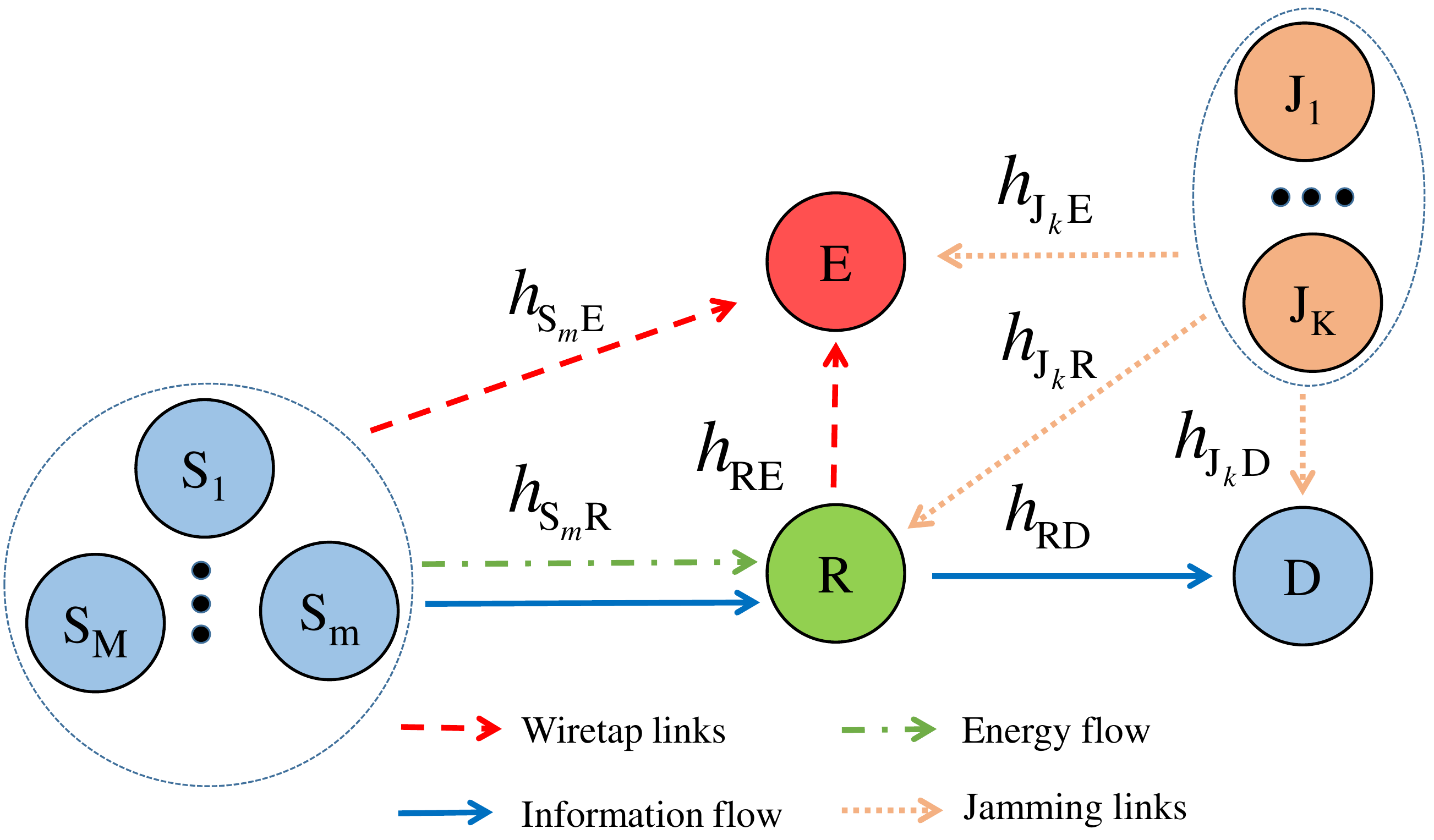} 
	\end{center}
	\caption{A SWIPT-based relay system with friendly jammers against attack from an eavesdropper.}
	\label{fig:1}
\end{figure*}

\begin{figure*} [t]
	\begin{center}
		\includegraphics[trim=0.2cm 0.2cm 0.3cm 0.2cm, clip=true, width=3.2in]{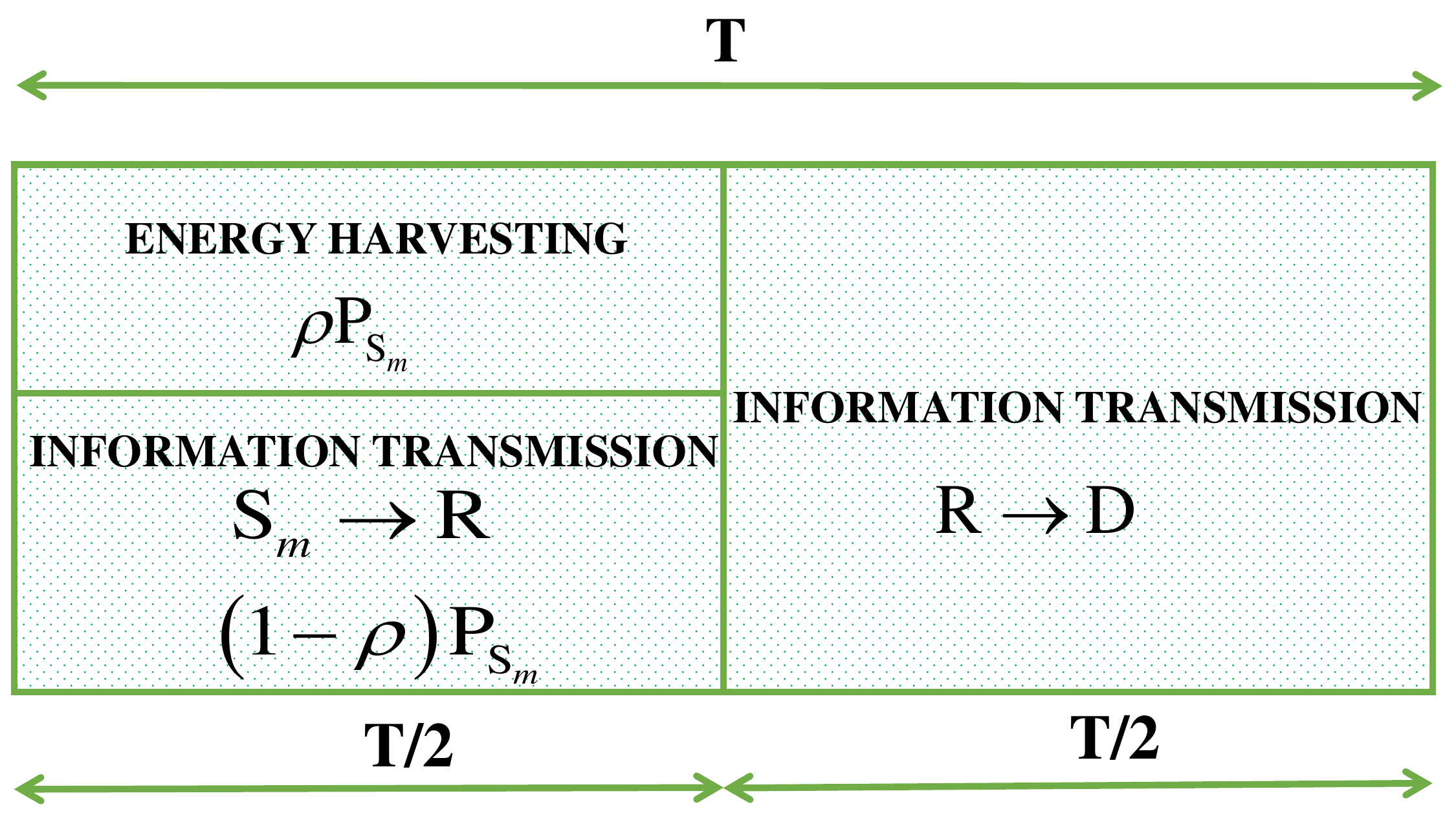} 
	\end{center}
	\caption{EH and IT processes with PS relaying protocol.}
	\label{fig:2}
\end{figure*}
Assuming the channel coefficient between any two nodes follows Rayleigh fading in which the channel is unchanged within a transmission block and to vary independently on different blocks. Let ${h_{{\rm{XY}}}}$ with ${\rm{XY}} \in \left\{ {{\rm{S_mR, RE, RD, JE, JR, JD}}} \right\}$ denote the channel  from $\rm X \to Y$, then the corresponding channel gains can be defined as
\begin{equation} \label{eq:gammaXY}
{\gamma _{{\rm{XY}}}} = {\left| {{h_{{\rm{XY}}}}} \right|^2}, {\rm{XY}} \in \left\{ {{\rm{S_mR, RE, RD, JE, JR, JD}}} \right\}.
\end{equation} Notice that The channel gains are assumed to be exponential random variables in which probability density function (PDF) and cumulative distribution function (CDF) are respectively represented as
\begin{align}
	\label{EQ1}
	{F_{X}}(x) &= 1 - \exp \left( { - \lambda x} \right),\\
	\label{EQ2}
	{f_{X}}(x) &= \frac{{\partial {F_{\rm{X}}}(x)}}{{\partial x}} = \lambda \exp ( - \lambda x).
\end{align}
where $\lambda$ is the rate parameter of the exponential random variable $X$. To take path-loss into account, we can model the parameters as follows:
\begin{align}
	\label{EQ3}
	{\lambda _{{\rm{XY}}}} = {\left( {{d_{{\rm{XY}}}}} \right)^\chi }.
\end{align}
where $\chi $ is the pathloss exponent and ${{d_{{\rm{XY}}}}}$ is distance between node X and Y.

Now, we discuss the adopted system model. As mention above, supposing source $\rm{S}_m$ is chosen to send its information and energy. During the first transmission phase, the received signal at the relay can be given by:
\begin{align}
	\label{EQ4}
	{y_{\rm{R}}} = \sqrt {1 - \rho } {h_{{{\rm{S}}_m}{\rm{R}}}}{x_{{{\rm{S}}_m}}} + \sum\limits_{k = 1}^K {{x_k}{h_{{{\rm{J}}_k}{\rm{R}}}}}  + {n_{\rm{R}}},
\end{align}
where ${x_{{{\rm{S}}_m}}}$ is the message transmitted with $\mathbb{E}\left\{ {{{\left| {{x_{{\rm{S}}_m}}} \right|}^2}} \right\} = {P_{\rm{S}}}$; ${x_k}$ is the artificial noise transmitted by the jammer ${{\rm{J}}_k}$, which satisfies  $\mathbb{E}\left\{ {{{\left| {{x_{{\rm{k}}}}} \right|}^2}} \right\} = {P_{\rm{J}}}$ , $\mathbb{E}\left\{  \cdot  \right\}$ denotes the expectation operation; ${n_{{{\rm{R}}}}}$ is the zero mean additive white Gaussian noise (AWGN) with variance $N_0$. Our proposed model only considers the multi-friendly jammers, which are just against the eavesdropper. Hence, relay R and destination D will know the jamming signals and cancel them in their received signals. So, the received signal at R can be rewritten as:
\begin{align}
	\label{EQ5}
	{y_{\rm{R}}} = \sqrt {1 - \rho } {h_{{{\rm{S}}_m}{\rm{R}}}}{x_{{{\rm{S}}_m}}} + {n_{\rm{R}}}.
\end{align}
Following the same methodology as in \cite{tin2021Sensorsperformance}, the energy harvesting in relay can be computed as 
\begin{align}
	\label{EQ6}
	{P_{\rm{R}}} = \frac{{{E_{\rm{R}}}}}{{T/2}} = \eta \rho {P_{{{\rm{S}}_m}}}{\gamma _{{{\rm{S}}_m}{\rm{R}}}},
\end{align}
where $0 < \eta  \le 1$ is energy conversion efficiency (which takes into account the energy loss by harvesting
circuits and also by decoding and processing circuits); $P_{{\rm {S}}_m}$ and $P_{\rm R}$ are the transmit powers of ${\rm S}_n$ and ${\rm R}$, respectively. $E_{\rm R}$ is the amount of the harvested energy at the relay and it can be calculated as $E_{\rm R} \triangleq \eta \rho {P_{{{\rm{S}}_m}}}{\gamma _{{{\rm{S}}_m}{\rm{R}}}} T/2$. The channel gain ${\gamma _{{{\rm{S}}_m}{\rm{R}}}}$ is defined as ${\gamma _{{{\rm{S}}_m}{\rm{R}}}} =  {\left| {{h_{{{\rm{S}}_m}{\rm{R}}}}} \right|^2}$, which is a special case of \eqref{eq:gammaXY}. Besides, $\rho \in [0,1]$ is the power splitting factor, where $\rho$ equals to zero (or one) signifies that total received RF signals used for EH (or information transmission).

As mentioned above, the received signal at the destination can be given as:
\begin{align}
	\label{EQ7}
	{y_{\rm{D}}} = {h_{{{\rm{R}}}{\rm{D}}}}{x_{{{\rm{R}}}}} + {n_{\rm{D}}},
\end{align}
where $n_{\rm{D}}$ is the zero mean AWGN with variance $N_0$.

In this paper, we consider the AF relaying protocol. Therefore, the signal transmitted by the relay is an amplified version of ${y_{\rm{R}}}$, which is denoted by a factor $\beta$ as follows
\begin{align}
	\label{EQ8}
	\beta  = \frac{{{x_{\rm{R}}}}}{{{y_{\rm{R}}}}} = \sqrt {\frac{{{P_{\rm{R}}}}}{{(1 - \rho ){P_{{{\rm{S}}_m}}}{{\left| {{h_{{{\rm{S}}_m}{\rm{R}}}}} \right|}^2} + {N_0}}}}  \approx \sqrt {\frac{{\eta \rho }}{{1 - \rho }}} .
\end{align}
From \eqref{EQ5} and \eqref{EQ8}, the received signal at D can be rewritten by
\begin{align}
	\label{EQ9}
	{y_{\rm{D}}} &= {h_{{\rm{RD}}}}{x_{\rm{R}}} + {n_{\rm{D}}} = {h_{{\rm{RD}}}}\beta {y_{\rm{R}}} + {n_{\rm{D}}} \notag\\
	&= {h_{{\rm{RD}}}}\beta \left[ {\sqrt {1 - \rho } {h_{{{\rm{S}}_m}{\rm{R}}}}{x_{{{\rm{S}}_m}}} + {n_{\rm{R}}}} \right] + {n_{\rm{D}}}\notag\\
	&= \underbrace {\sqrt {1 - \rho } {h_{{\rm{RD}}}}{h_{{{\rm{S}}_m}{\rm{R}}}}\beta {x_{{{\rm{S}}_m}}}}_{\rm signal} + \underbrace {{h_{{\rm{RD}}}}\beta {n_{\rm{R}}} + {n_{\rm{D}}}}_{\rm noise}.
\end{align}

Hence, the signal to noise ratio (SNR) at D in this phase can be obtained by:
\begin{align}
	\label{EQ10}
	{\gamma _{\rm{D}}} = \frac{{{\rm E}\left\{ {{{\left| {\rm signal} \right|}^2}} \right\}}}{{{\rm E}\left\{ {{{\left| {\rm noise} \right|}^2}} \right\}}} = \frac{{(1 - \rho ){P_{{{\rm{S}}_m}}}{\gamma _{{{\rm{S}}_m}{\rm{R}}}}{\gamma _{{\rm{RD}}}}{\beta ^2}}}{{{\gamma _{{\rm{RD}}}}{\beta ^2}{N_0} + {N_0}}}.
\end{align}

After some algebraic manipulations, we have:
\begin{align}
	\label{EQ11}
	{\gamma _{\rm{D}}} = \frac{{\eta \rho (1 - \rho )\Psi {\gamma _{{{\rm{S}}_m}{\rm{R}}}}{\gamma _{{\rm{RD}}}}}}{{\eta \rho {\gamma _{{\rm{RD}}}} + (1 - \rho )}},
\end{align}
where $\Psi  = \frac{{{P_{\rm{S}}}}}{{{N_0}}}.$ The achievable rate of system can be claimed by:
\begin{align}
	\label{EQ12}
	{C_{{\rm{AF}}}} = \frac{1}{2}{\log _2}\left( {1 + {\gamma _{\rm{D}}}} \right).
\end{align}
Taking into account the impact of eavesdropper E, at the first time slot the chosen source ${\rm{S}_m}$ will broadcast its signal to relay and E can overhear this information. To prevent E from eavesdropping legitimate information, jammers will send the jamming signals to E, so the received signal at E can be expressed by:
\begin{align}
	\label{EQ13}
	y_{\rm{E}}^1 = {h_{{{\rm{S}}_m}{\rm{E}}}}{x_{{{\rm{S}}_m}}} +  {h_{{\rm{RE}}}}{x_{\rm{R}}} + \sum\limits_{k = 1}^K {{x_k}{h_{{{\rm{J}}_k}{\rm{E}}}}}  + {n_{\rm{E}}},
\end{align}
where ${n_{\rm{E}}}^1$ is the zero mean AWGN with variance $N_0$. Based on \eqref{EQ13}, the SNR at the E during first time slot is given as:
\begin{align} \label{eq:gammaE1}
	\gamma _{\rm{E}}^1 = \frac{{{P_{{{\rm{S}}_m}}}{\gamma _{{{\rm{S}}_m}{\rm{E}}}}}}{{{P_{\rm{J}}}\sum\limits_{k = 1}^K {{\gamma _{{{\rm{J}}_k}{\rm{E}}}}}  + {N_0}}} = \frac{{\Psi {\gamma _{{{\rm{S}}_m}{\rm{E}}}}}}{{\Phi \Xi  + 1}} \approx \frac{{\Psi {\gamma _{{{\rm{S}}_m}{\rm{E}}}}}}{{\Phi \Xi }},
\end{align}
where $\Phi  = \frac{{{P_{\rm{J}}}}}{{{N_0}}}$ and $\Xi  = \sum\limits_{k = 1}^K {{\gamma _{{{\rm{J}}_k}{\rm{E}}}}} $. Here, please note that after received the information from chosen source ${{\rm{S}}_m}$, R will amplify this information to D and also to E with the same amplification factor $\beta$. Therefore, the received signal at E can be re-expressed by:
\begin{align}
	\label{EQ14}
		&{y_{\rm{E}}^2}  = {h_{{\rm{RE}}}}\beta \left[ {\sqrt {1 - \rho } {h_{{{\rm{S}}_m}{\rm{R}}}}{x_{{{\rm{S}}_m}}} + {n_{\rm{R}}}} \right] + \sum\limits_{k = 1}^K {{x_k}{h_{{{\rm{J}}_k}{\rm{R}}}}}  + {n_{\rm{E}}^2}\notag\\
		&= \underbrace {\sqrt {1 - \rho } {h_{{{\rm{S}}_m}{\rm{R}}}}{h_{{\rm{RE}}}}\beta {x_{{{\rm{S}}_m}}}}_{signal} + \underbrace {{h_{{\rm{RE}}}}\beta {n_{\rm{R}}} + \sum\limits_{k = 1}^K {{x_k}{h_{{{\rm{J}}_k}{\rm{E}}}}}  + {n_{\rm{E}}^2}}_{noise},
\end{align}
where $n_{\rm{E}}^2$ is the zero mean AWGN with variance $N_0$. Hence, the received SNR at E during second time slot can be obtained as:
\begin{align}
	\label{EQ15}
	{\gamma _{\rm{E}}^2} &= \frac{{(1 - \rho ){\gamma _{{{\rm{S}}_m}{\rm{R}}}}{\gamma _{{\rm{RE}}}}{\beta ^2}{P_{{{\rm{S}}_m}}}}}{{{\gamma _{{\rm{RE}}}}{\beta ^2}{N_0} + {P_{\rm{J}}}\sum\limits_{k = 1}^K {{\gamma _{{{\rm{J}}_k}{\rm{E}}}}}  + {N_0}}} \notag \\ &= \frac{{\eta \rho (1 - \rho ){\gamma _{{{\rm{S}}_m}{\rm{R}}}}{\gamma _{{\rm{RE}}}}\Psi }}{{\eta \rho {\gamma _{{\rm{RE}}}} + \left( {1 - \rho } \right)\Phi \Xi  + \left( {1 - \rho } \right)}}.
\end{align}
Finally, the eavesdropper can apply selection combining technique (SC), thus, the achievable rate at E can be given by:
\begin{align}
	\label{EQ16}
	{C_{\rm{E}}} = \frac{1}{2}{\log _2}\left( {1 + {\gamma _{\rm{E}}}} \right),
\end{align}
where ${\gamma _{\rm{E}}} = \max \left( {\gamma _{\rm{E}}^1,\gamma _{\rm{E}}^2} \right).$

\begin{remark}
The best source $\rm{S}_m$ would be selected for the purpose of enhancement the tranmission performance. Mathematically speaking, we can write as followings
\begin{align}
	\label{EQ17}
	n &\triangleq \arg \underbrace {\max }_{m = 1,2,...,M}\left\{ {{\gamma _{{{\rm{S}}_m}{\rm{R}}}}} \right\} \notag\\ \Leftrightarrow {\gamma _{{{\rm{S}}_m}{\rm{R}}}} &\triangleq \underbrace {\max }_{m = 1,2,...,M}\left\{ {{\gamma _{{{\rm{S}}_m}{\rm{R}}}}} \right\}.
\end{align}

Assume that the links from M sources to the relay R are identical independent distribution (i.i.d), and hence CDF and PDF of ${\gamma _{{{\rm{S}}_m}{\rm{R}}}}$ can be formulated as, respectively.
\begin{align}
	\label{EQ18}
		{F_{{\gamma _{{{\rm{S}}_m}{\rm{R}}}}}}(x) &= \Pr \left( {{\gamma _{{{\rm{S}}_m}{\rm{R}}}} < x} \right) \notag\\&= \Pr \left( {{\gamma _{{{\rm{S}}_m}{\rm{R}}}} < x,\forall m = 1,2,...,M} \right) \notag\\
		&= \prod\limits_{m = 1}^M {{F_{{\gamma _{{{\rm{S}}_m}{\rm{R}}}}}}(x)}  = {\left[ {1 - {e^{ - {\lambda _{{{\rm{S}}_m}{\rm{R}}}}x}}} \right]^M} \notag\\ &= 1 + \sum\limits_{b = 1}^M {{{\left( { - 1} \right)}^b}} C_M^b\exp \left( { - b{\lambda _{{\rm{SR}}}}x} \right),
\end{align}
\begin{align}
	\label{EQ19}
		{f_{{\gamma _{{{\rm{S}}_m}{\rm{R}}}}}}(x) = {\lambda _{{\rm{SR}}}}\sum\limits_{b = 1}^{M - 1} {{{( - 1)}^b}} C_{M - 1}^b\exp \left( { - (b + 1){\lambda _{{\rm{SR}}}}x} \right),
\end{align}
where $C_M^b = \frac{{M!}}{{b!(M - b)!}}.$
\end{remark}

\begin{remark}
Because $\Xi $ is a summation of $K$ i.i.d. exponential random variables, its PDF can be given by \cite{49}
\begin{align}
	\label{EQ20}
	{f_\Xi }(x) = \frac{{{{({\lambda _{{\rm{JE}}}})}^K}}}{{(K - 1)!}}{x^{K - 1}}\exp ( - {\lambda _{{\rm{JE}}}}x).
\end{align}
\end{remark}
\section{Performance Analysis}
\label{sec:3}
This section provides the mathematical analysis of the outage probability (OP)  and Intercept probability (IP) to provide further insight into the two-hop data transmission for SWIPT-based HD AF relay networks in two cases, i.e., dynamic power splitting relaying (DPSR) and static power splitting relaying (SPSR).
\subsection{SPSR case}
\subsubsection{Outage Probability (OP) Analysis}
In this section, the OP of the SWIPT-aided HD AF relaying system over Rayleigh fading channels is derived. Specifically, it can be calculated as:
\begin{align}
	\label{EQ21}
	{\rm{OP}} = \Pr \left( {{C_{AF}} < {C_{th}}} \right) = \Pr \left( {{\gamma _{D}} < {\gamma _{th}}} \right),
\end{align}
where ${\gamma _{th}} = {2^{2{C_{th}}}} - 1$ is the signal-to-noise ratio (SNR) threshold of system and ${C_{th}}$ is the target rate. The closed-form expression for the OP in \eqref{EQ21} is given as follows.
\begin{theorem} \label{Theorem:Outage1}
	In static power splitting-based relaying, the closed-form expression of the OP can be given as
\begin{align}
	\label{eq:Theorem_1}
		{\rm{O}}{{\rm{P}}_{{\rm{SPSR}}}} &= 1 + 2\sum\limits_{b = 1}^M {{{( - 1)}^b}C_M^b} \exp \left( { - \frac{{b{\lambda _{{\rm{SR}}}}{\lambda _{th}}}}{{(1 - \rho )\Psi }}} \right) \notag\\ &\times \sqrt {\frac{{b{\lambda _{{\rm{SR}}}}{\lambda _{{\rm{RD}}}}{\gamma _{th}}}}{{\eta \rho \Psi }}} {K_1}\left( {2\sqrt {\frac{{b{\lambda _{{\rm{SR}}}}{\lambda _{{\rm{RD}}}}{\gamma _{th}}}}{{\eta \rho \Psi }}} } \right),
\end{align}
where ${K_v}\left(  \cdot  \right)$ is the modified Bessel function of the second kind and $v$-th order.
\end{theorem}

\begin{proof}
The detailed proof is available in Appendix~\ref{Appendix:Outage1}.
\end{proof}

\subsubsection{Intercept Probability (IP) Analysis}
Destination D  will be intercepted if E can successfully wiretap signal, i.e. ${C_E} \ge {C_{th}}$. Therefore, the IP can be defined as \cite{9131891}
\begin{align}
	\label{EQ25}
	{\rm{IP = Pr}}\left( {{C_{\rm{E}}} \ge {C_{th}}} \right) = \Pr \left( {{\gamma _{\rm{E}}} \ge {\gamma _{th}}} \right).
\end{align}

\begin{theorem} \label{Theorem:IP}
	The closed-form expression of the IP in this case can be derived by:
\begin{align}
		\label{eq:Theorem_2}
		{\rm{I}}{{\rm{P}}_{{\rm{SPSR}}}}  &= {\left( {\frac{{{\lambda _{{\rm{JE}}}}}}{{{{\tilde \lambda }_{{\rm{JE}}}}}}} \right)^K} + \sum\limits_{t = 0}^\infty  \sum\limits_{b = 1}^M {\frac{{{{\left( { - 1} \right)}^{b + t}}C_M^b\Gamma \left( {t + K} \right)}}{{t!{\Phi ^{t + K}}}}} \notag\\ & \times \exp \left( { - \frac{{b{\lambda _{{\rm{SR}}}}{\gamma _{th}}}}{{{\rho _1}\Psi }}} \right) \notag\\
		&\times \frac{{{{\left( {{\lambda _{{\rm{JE}}}}} \right)}^K}}}{{\left( {K - 1} \right)!}} \times G_{1,3}^{3,0}\left( {\frac{{b{\lambda _{{\rm{SR}}}}{\lambda _{{\rm{RE}}}}{\gamma _{th}}}}{{\eta \rho \Psi }}\left| \begin{array}{l}
				0\\
				- t - K,1,0	\end{array} \right.} \right) \notag\\& \times \left[ {{{\left( {{{\tilde \lambda }_{{\rm{JE}}}}} \right)}^t} - {{\left( {{\lambda _{{\rm{JE}}}}} \right)}^t}} \right],
\end{align}
where $\Gamma \left(  \cdot  \right)$ is the Gamma function and ${\rm{G}}_{p,q}^{m,n}\left( {z\left| \begin{array}{l}
		{a_1},...,{a_p}\\
		{b_1},...,{b_q}
	\end{array} \right.} \right)$ is the Meijer G-function.
\end{theorem}
\begin{proof}
The detailed proof is available in Appendix~\ref{Appendix:IP}.
\end{proof}

\subsection{DPSR case} \label{Sec:DPSRcase}
In this section, we would like to find the optimal power splitting factor, i.e., $\rho^\star$ to maximize the achievable rate $C_{\rm AF}$. By observing \eqref{EQ12}, we have $\max \left( {{C_{{\rm{AF}}}}} \right) \Leftrightarrow \max \left( {{\gamma _{\rm{D}}}} \right)$. It is easy to prove that $\frac{{{\partial ^2}{\gamma _{\rm{D}}}}}{{{\partial ^2}\rho }}$ is negative for all $0 < \rho  < 1$. Hence, we conclude that ${\gamma _{\rm{D}}}$ is a concave function of $\rho$. We can find the value of $\rho$ to maximize ${\gamma _{\rm{D}}}$ by differentiating ${\gamma _{\rm{D}}}$ concerning $\rho$ and then equate it to zero. After doing some algebraic calculations, $\rho^\star$ can be given as ${\rho ^*} = \frac{1}{{1 + \left| {{h_{{\rm{RD}}}}} \right|\sqrt \eta  }}$ or ${\rho ^*} = \frac{1}{{1 - \left| {{h_{{\rm{RD}}}}} \right|\sqrt \eta  }}$. Because of ${\rho ^*} = \frac{1}{{1 - \left| {{h_{{\rm{RD}}}}} \right|\sqrt \eta  }}$ results in the value of ${\rho ^*} > 1$ or ${\rho ^*} < 0$. Therefore, ${\rho ^*} = \frac{1}{{1 + \left| {{h_{{\rm{RD}}}}} \right|\sqrt \eta  }}$ is selected as the optimal solution. 

\begin{theorem} \label{Theorem:Outage2}
The closed-form expression of the OP in this case can be derived by:
\begin{multline}
	\vspace{-0.1cm}
	\label{eq:Theorem_3}
	{\rm{O}}{{\rm{P}}_{{\rm{DPSR}}}} = 1 + \sum\limits_{t = 0}^\infty  \sum\limits_{b = 1}^M \frac{{{{( - 1)}^{t + b}}C_M^b{2^{t + 1}}}}{{t!}} \\ \times {{\left( {\frac{{{\lambda _{{\rm{RD}}}}}}{\eta }} \right)}^{t/4 + 1/2}} {{\left( {\frac{{b{\lambda _{{\rm{SR}}}}{\gamma _{th}}}}{\Psi }} \right)}^{3t/4 + 1/2}}  \\
	\times \exp \left( { - \frac{{b{\lambda _{{\rm{SR}}}}{\gamma _{th}}}}{\Psi }} \right) \times {K_{ - t/2 + 1}}\left( {2\sqrt {\frac{{b{\lambda _{{\rm{SR}}}}{\lambda _{{\rm{RD}}}}{\gamma _{th}}}}{{\eta \Psi }}} } \right),
\end{multline}
where ${K_v}\left(  \cdot  \right)$ is the modified Bessel function of the second kind and ${v^{th}}$ order.
\end{theorem}
\begin{proof}
The detailed proof is available in Appendix~\ref{Appendix:Outage2}.
\end{proof}
\subsubsection{IP Analysis}
\begin{theorem}\label{Theorem:IPv2}
	In dynamic power splitting-based relaying, the exact integral-form expression of the IP is derived as:
\begin{align}
	\label{eq:Theorem_4}
		{\rm{I}}{{\rm{P}}_{{\rm{DPSR}}}} &= {\left( {\frac{{{\lambda _{{\rm{JE}}}}}}{{{{\tilde \lambda }_{{\rm{JE}}}}}}} \right)^K} + 2\sum\limits_{b = 1}^M \frac{{{{\left( { - 1} \right)}^b}C_M^b{{\left( {{\lambda _{{\rm{JE}}}}} \right)}^K}}}{{\left( {K - 1} \right)!}} \notag \\
 &\times \exp \left( { - \frac{{b{\lambda _{{\rm{SR}}}}{\gamma _{th}}}}{\Psi }} \right) \times \sqrt {\frac{{b{\lambda _{{\rm{SR}}}}{\lambda _{{\rm{RE}}}}{\gamma _{th}}}}{{\eta \Psi }}}  \notag\\
		&\times \Bigg[ \int\limits_0^\infty  {{x^{K - 1}}\exp \left( { - {{\tilde \lambda }_{{\rm{JE}}}}x} \right)\Delta (\omega )\sqrt {\left( {\Phi x + 1} \right)} } dx \notag\\
		&- \int\limits_0^\infty  {{x^{K - 1}}\exp \left(  - {\lambda _{{\rm{JE}}}}x \right)\Delta (\omega )\sqrt {\left( {\Phi x + 1} \right)} } dx \Bigg].
\end{align}
\end{theorem}
We have obtained the analytical results in Theorems~\ref{Theorem:Outage1}--\ref{Theorem:IPv2}, which are independent of small-scale fading coefficients and only based on the statistical CSI. The obtained analyses work for a long period of time and reduce the computational complexity in evaluating the OP and IP of the network.  
\begin{remark}
This paper considers scenarios where the jammers are closer to the eavesdropper than the relay. Hence, the RF transmission from friendly jammers was neglected for the sake of simplicity. Once the information of jammers is carefully exploited, both eavesdropper (E) and destination (D) will get benefits to enhance the system performance. A framework using the RF transmission from friendly jammers shares a similar analytical methodology as initially established in this paper but with more complicated in detail. This interesting extension is left for future work.
\end{remark} 
\begin{table*}[t]
	\caption{Simulation parameters}
	\label{Table_1}
	\centering
	\setlength{\tabcolsep}{5pt}
	\begin{tabular}{|>{\centering\arraybackslash} m{2cm} | >{\centering\arraybackslash} m{5cm}| >{\centering\arraybackslash} m{3cm}| >{\centering\arraybackslash} m{3cm}|}
		\hline \hline
		\textbf{Symbol} & \textbf{Parameter name} & \textbf{Fixed value} & \textbf{Varying range} \\  \hline\hline
		${{\rm{C}}_{th}}$&    Source rate&    0.25; 0.5 bps/Hz& none\\
		$\eta$& Energy harvesting efficiency& 0.8& none\\
		$\rho$& Power splitting factor& 0.225; 0.325; 0.5; 0.875; 0.915& 0 to 1\\
		$\lambda_{\rm{S_mR}}$& Rate parameter of ${\left| {{h_{{{\rm{S}}_m}{\rm{R}}}}} \right|^2}$ & 0.1768& none\\
		$\lambda_{\rm{RD}}$& Rate parameter of ${\left| {{h_{{\rm{RD}}}}} \right|^2}$ & 0.1768& none \\ 
		$\lambda_{\rm{RE}}$& Rate parameter of ${\left| {{h_{{\rm{RE}}}}} \right|^2}$ & 2.7557 ($\rm{S}1$); 1.3216 ($\rm{S}2$) & none\\ 
		$\lambda_{\rm{JE}}$& Rate parameter of ${\left| {{h_{{\rm{JE}}}}} \right|^2}$ & 0.1768 (S1); 1 (S2)& none\\ 
		$\lambda_{\rm{SE}}$ & Rate parameter of ${\left| {{h_{{\rm{SE}}}}} \right|^2}$ &  3.1434 (S1); 1 (S2) & none\\  
		$\Psi$& Transmit-power-to-noise-ratio of source & 2 dB& -5 to 15 (dB) \\ 
		$\Phi $& Transmit-power-to-noise-ratio of jammer & -1;1 dB& none \\ 
		$M$ & Number of source node(s) & 2; 3 & 1 to 10 \\ 
		$K$ & Number of jammer(s) & 1 & 1 to 10 \\ \hline
		\hline
	\end{tabular}
\end{table*}
\begin{figure*}[t]
	\centering
	\begin{minipage}{.48\textwidth}
		\centering
		\includegraphics[width=9.5cm,height=8cm]{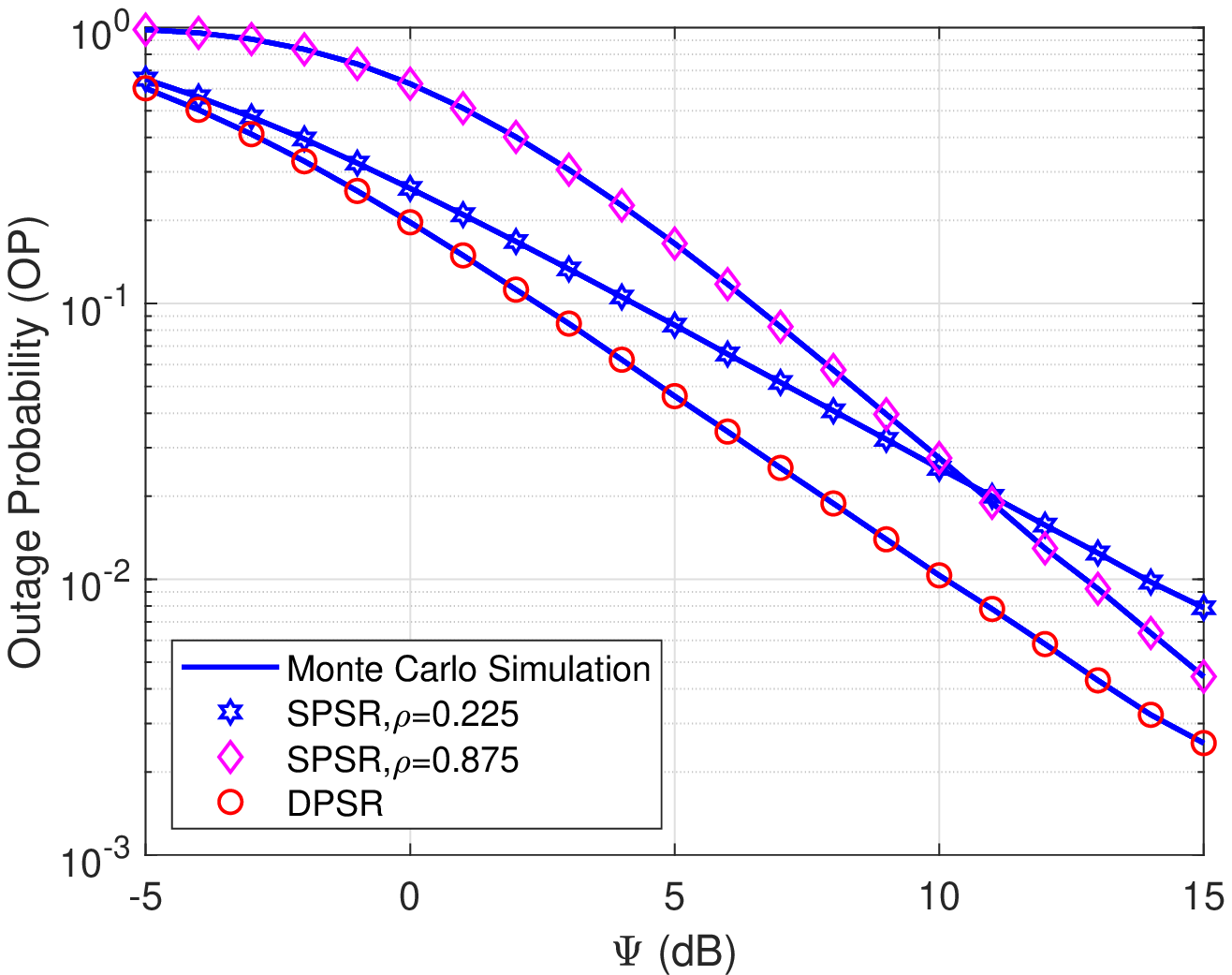}
		\caption{OP vs $\Psi$ with ${C_{th}}$=0.5 bps/Hz, $\eta$=0.8, $M$=2.}
		\label{fig:3}
	\end{minipage} \hfill
	\begin{minipage}{.48\textwidth}
		\centering
		\includegraphics[width=9.5cm,height=8cm]{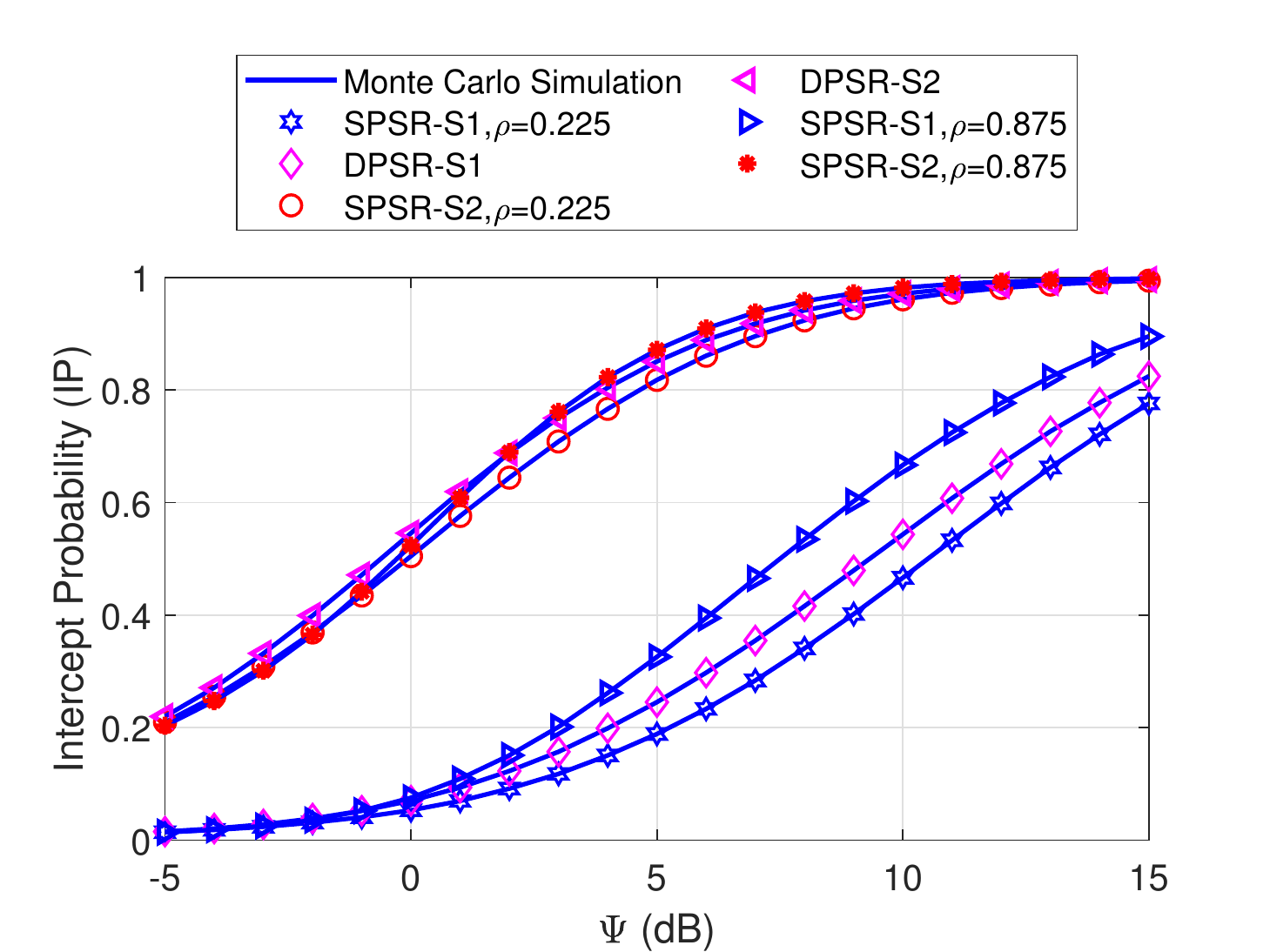}
		\caption{IP vs $\Psi$ with ${C_{th}}$=0.5 bps/Hz, $\eta$=0.8, $M$=2, $K$=1 and $\Phi$ = 1~dB.}
		\label{fig:4}
	\end{minipage}
\vspace{-0.5cm}
\end{figure*}
\section{Simulation Results}
\label{Sec:Num}
In this section, Monte-Carlo simulations are provided to validate the theoretical expressions and the impacts of various
parameters on the system performance. To claim the OP and IP for the proposed schemes, we perform $5\times10^6$ independent trials, and the channel coefficients are randomly generated as Rayleigh fading in each trial. The settings of simulation parameters are detailed in Table~\ref{Table_1}. In particular, we provide two scenarios in other to investigate the system performance corresponding to different node deployments. In the first scenario (S1), sources are located around $(0,0)$, the relay is located at $(0.5,0)$, the destination is located at $(1,0)$, and the eavesdropper is located closer to the relay at $(0.5, 1.5)$. In the second scenario (S2), sources, relay, and destination are kept the same locations, while the eavesdropper is located closer to the sources at $(0,1)$. The average channel gains, $\mathbb{E}\{\gamma_{\rm{XY}}\}$ are computed by utilizing \eqref{eq:gammaXY}, where the propagation distances are obtained from the above setting. Furthermore, the rate parameters $\lambda_{\rm{S_mR}}$, $\lambda_{\rm{RD}}$, $\lambda_{\rm{JE}}$, $\lambda_{\rm{RE}}$, and $\lambda_{\rm{SE}}$ are given in Table~\ref{Table_1}. The system performance metrics comprising the OP and IP are evaluated as a functions of different parameters.

\begin{figure*}[t]
	\label{fig:5}
	\centering
	\begin{minipage}{.48\textwidth}
		\centering
		\includegraphics[width=9.5cm,height=8cm]{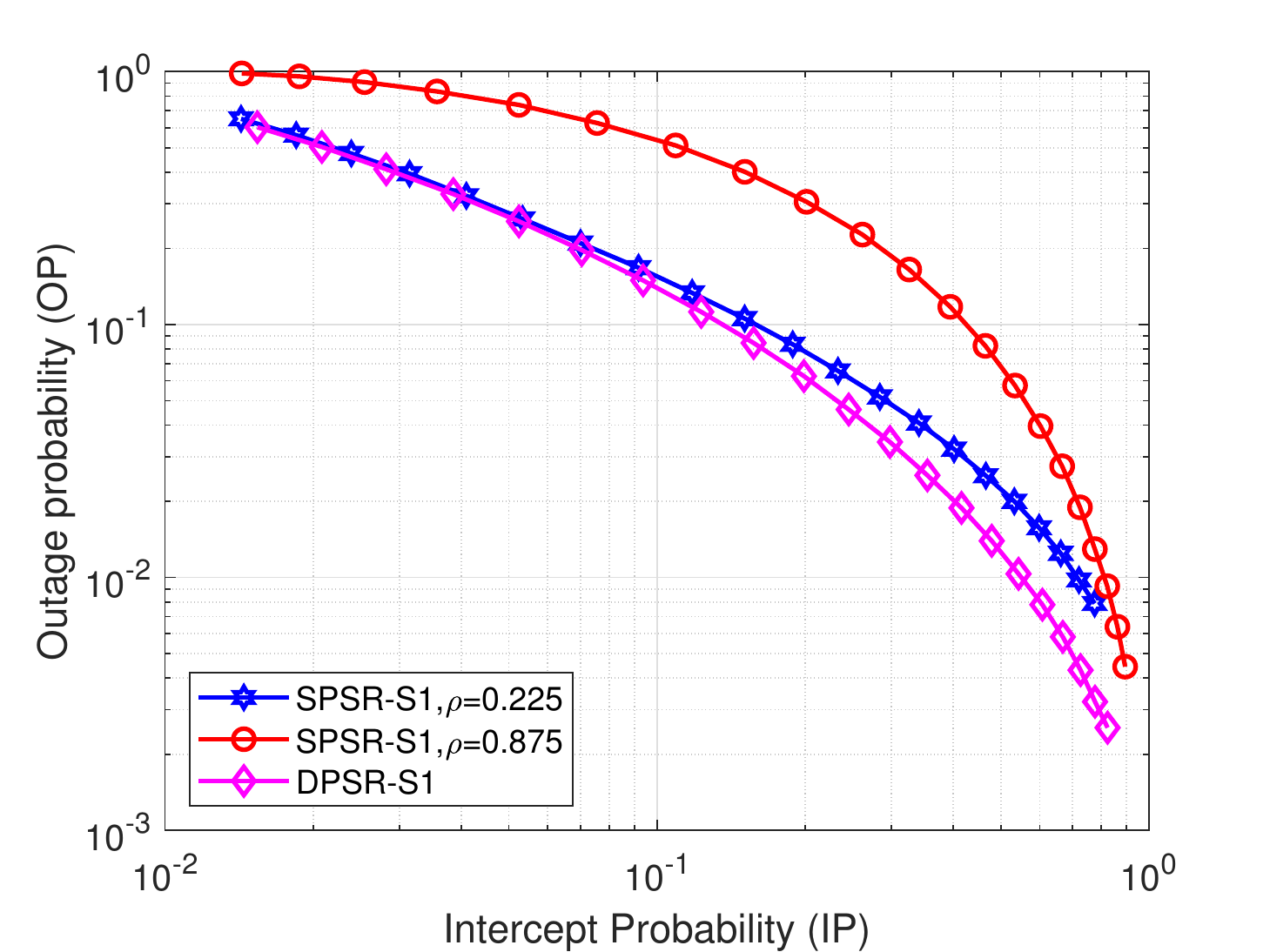}
		\caption{OP vs IP in scenario 1 with ${C_{th}}$=0.5 bps/Hz, $\eta$=0.8, $M$=2, $K$=1, $\Phi$=1 dB and $\Psi  \in \left[ { - 5,15\,dB} \right]$.}
		\label{fig:5a}
	\end{minipage} \hfill
	\begin{minipage}{.48\textwidth}
		\centering
		\includegraphics[width=9.5cm,height=8cm]{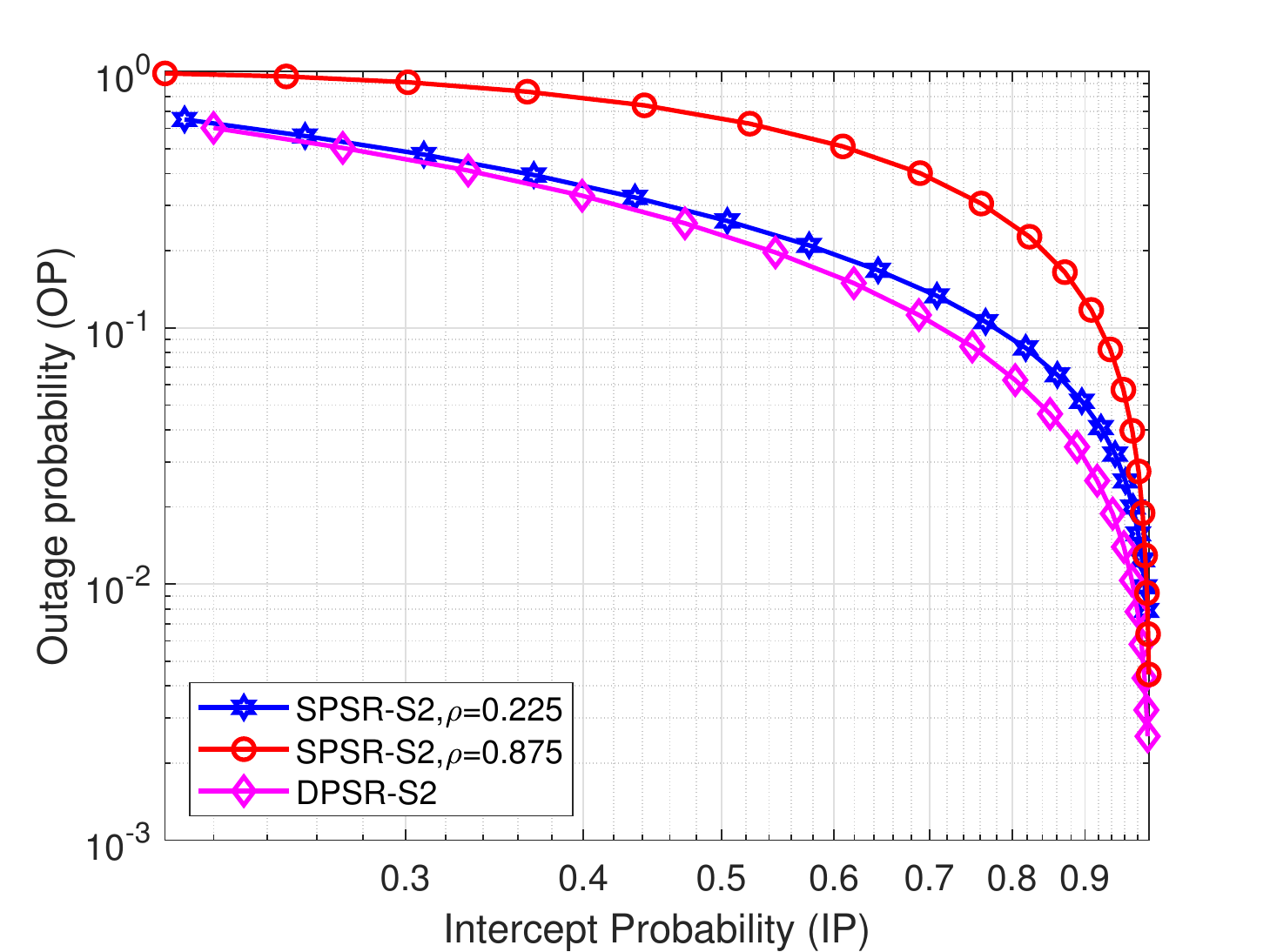}
		\caption{OP vs IP in scenario 2 with ${C_{th}}$=0.5 bps/Hz, $\eta$=0.8, $M$=2, $K$=1, $\Phi$=1 dB and $\Psi  \in \left[ { - 5,15\,dB} \right]$.}
		\label{fig:5b}
	\end{minipage}
    \vspace{-0.5cm}
\end{figure*}
\begin{figure} [t]
	\begin{center}
		\includegraphics[width=9.5cm,height=8cm]{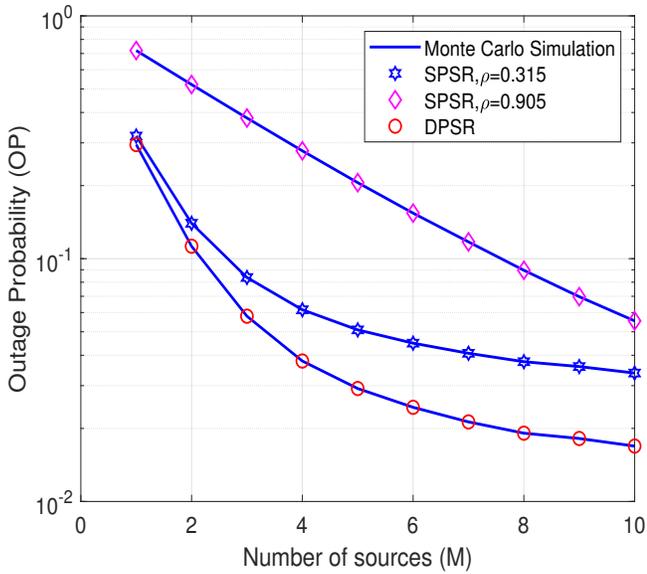} 
	\end{center}
	\caption{OP vs number of sources (M) with ${C_{th}}$=0.5 bps/Hz, $\eta$=0.8 and $\Psi$=2 dB.}
	\label{fig:6}
	\vspace{-0.5cm}
\end{figure}
\begin{figure} [t]
	\begin{center}
		\includegraphics[width=9.5cm,height=8cm]{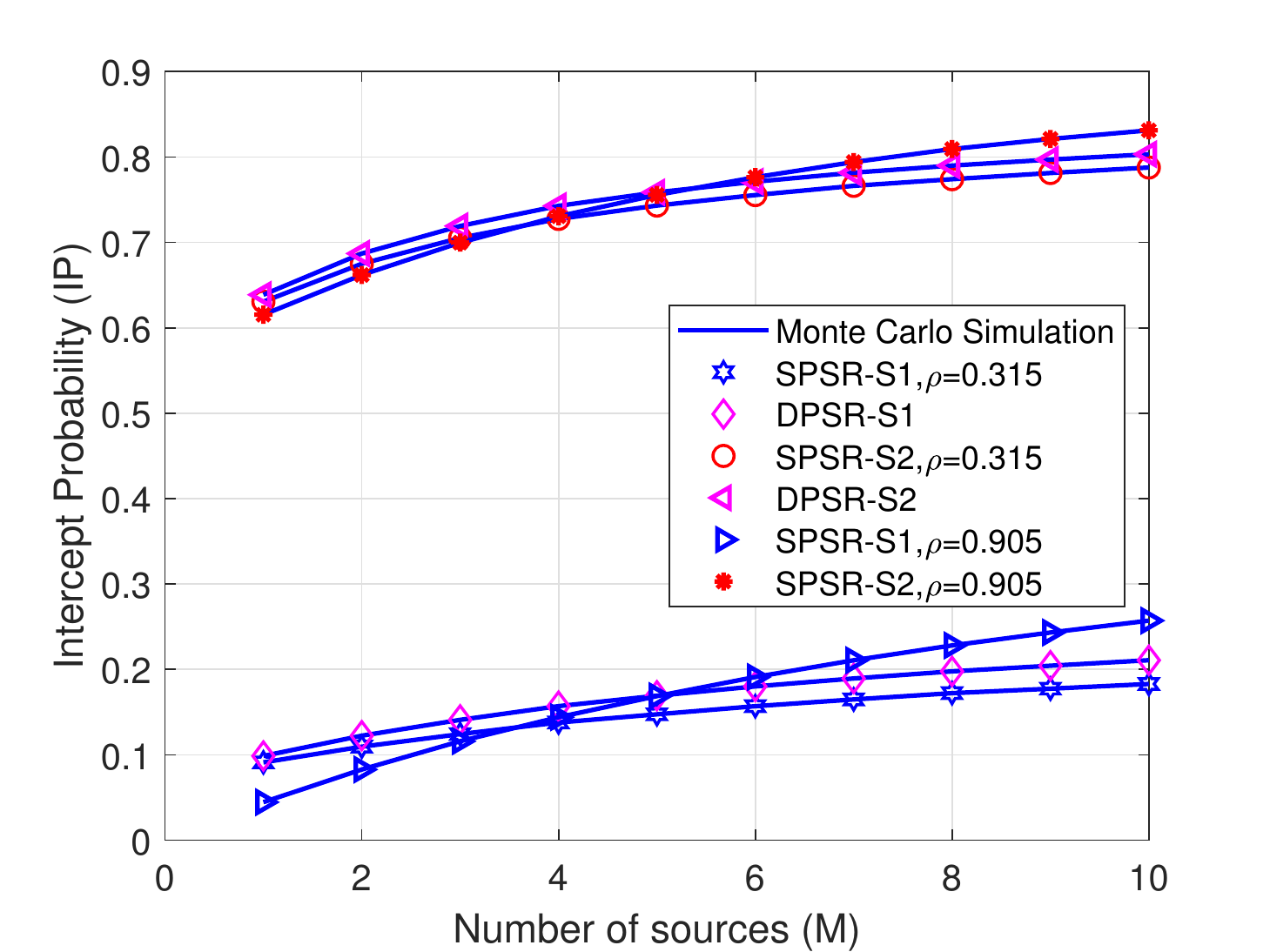} 
	\end{center}
	\caption{IP vs number of sources (M) with ${C_{th}}$=0.5 bps/Hz, $\eta$=0.8, $K$=1, $\Phi$=1 dB and $\Psi$=2 dB.}
	\label{fig:7}
	\vspace{-0.5cm}
\end{figure}
In Figs. \ref{fig:3} and \ref{fig:4}, we show the impact of $\Psi$ on OP and IP, where $\eta=$ 0.8, ${C_{th}}$= 0.5 bps/Hz, and $M = 2$. In Fig.~\ref{fig:3}, we compare DPSR with SPSR in OP analysis, whereas the SPSR is considered in two modes with $\rho$ equals 0.225 and 0.875, respectively. First, it is easy to observe that DPSR obtains better OP results than SPSR methods. Specifically, when $\Psi$ = 15 dB, the OP of DPSR approximately reach to $10^{-2.7}$, while the SPSBR with $\rho=0.225$ and $\rho=0.875$ impose 0.0025 and 0.0079, respectively. This is because the DPSR scheme aims to maximize the system capacity, thus it can improve the outage performance while the SPSR scheme always selects a fixed value of power splitting factor $\rho$. Second, the higher the $\Psi$ value is, the better the OP can be obtained. It can be explained by the fact that the higher $\Psi$ value means transmit power of source S is assigned more, which is defined in Eq. \eqref{eq:Theorem_1}. In Fig. \ref{fig:4}, the IP in both the DPSR and SPSR cases is studied, where $\eta=$ 0.8, ${C_{th}}$= 0.5 bps/Hz, M = 2, $K$=1 and $\Phi$=1 dB. It can be seen that the intercept performance increases with a higher value of $\Psi$. It is anticipated because the eavesdropper has more ability to overhear the message with a higher source transmit power $\rm{S}_m$. When the value of $\Psi$ is large enough, the IP in all schemes can converge to 1. An interesting thing in Fig. \ref{fig:4} is that the SPSR in second scenario obtains better IP performance as compared to the case in the first scenario. When $- 5 \le \Psi  \le 1$, the IP value of SPSR with $\rho$ = 0.875 is better than that of the SPSR with $\rho$ = 0.225. Otherwise, when $\Psi  >  - 1$ dB, the IP value of the SPSR with $\rho$=0.875 is higher than that of SPSR with $\rho$=0.225.  As $ - 5 \le \Psi  \le 0$ dB, the DPSR does not play a significant role compared with the other benchmarks since the power splitting factor $\rho$ is only optimized for the received signal strength of a specific legitimate user. For eavesdropper,  the IP  gets higher as the power splitting factor $\rho$ gets larger, which is a consequence of the high transmit power of relay R. This observation is applied for all the considered benchmarks.
In Figs. \ref{fig:3} and  \ref{fig:4}, simulations agree with the analytical values, which confirms our mathematical derivations' correctness.

In Figs.~\ref{fig:5a} and \ref{fig:5b}, we investigate security-reliability trade-off for S1 and S2 with the same parameter as in Figs. \ref{fig:3} and \ref{fig:4}. Nonetheless, these results demonstrate the benefits of our theoretical analysis. As observed in these figures, with increasing IP, the OP decrease and vice versa, indicating a trade-off between security and reliability. Figs. \ref{fig:5a} and \ref{fig:5b} also shows that the DPSR method can provide the best IP and OP value. However, this requires us to trade-off the choice that if we want the system obtains a better outage performance, it will be easier for the eavesdropper to steal information, on the contrary, if we want to restrict eavesdropping, we also reduce the outage performance.

Figs. \ref{fig:6} and \ref{fig:7} study the impact of the number of sources M on the OP and IP, respectively. Herein, simulation parameters are set as ${C_{th}}$=0.5 bps/Hz, $\eta$=0.8, $\Psi$=2 dB, $K$=1 and $\Phi$=1 dB. Fig. \ref{fig:6} shows a better OP as the number of sources increases, i.e., M varies from 1 to 10.  It is because, by increasing M, we will have more options to transmit information to the destination and choose the best source $\rm{S}_m$. Moreover, the OP of DPSR is still better than that of other SPSR cases. However, as shown in Fig. \ref{fig:7}, the IP will also increase by using a larger number of sources. Once again, this shows the trade-off between security and reliability.
\begin{figure*}[t]
	\centering
	\begin{minipage}{.48\textwidth}
		\centering
		\includegraphics[width=9.5cm,height=8cm]{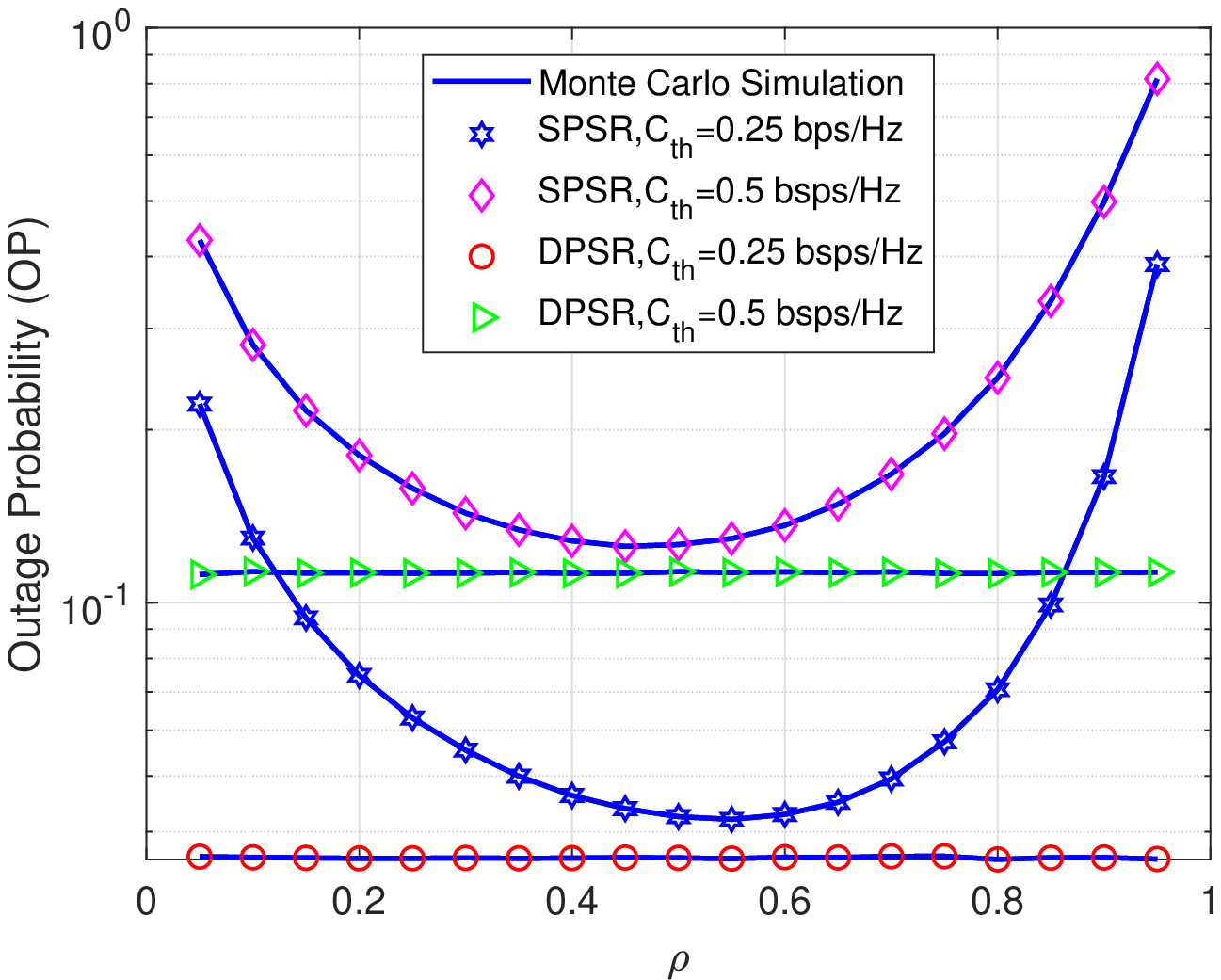}
		\caption{OP vs $\rho$ with $\eta$=0.8, $M$=2, $\Psi$=2 dB, $\Phi$=1 dB and K=1. }
		\label{fig:9}
	\end{minipage} \hfill
	\begin{minipage}{.48\textwidth}
		\centering		
		\includegraphics[width=9.5cm,height=8cm]{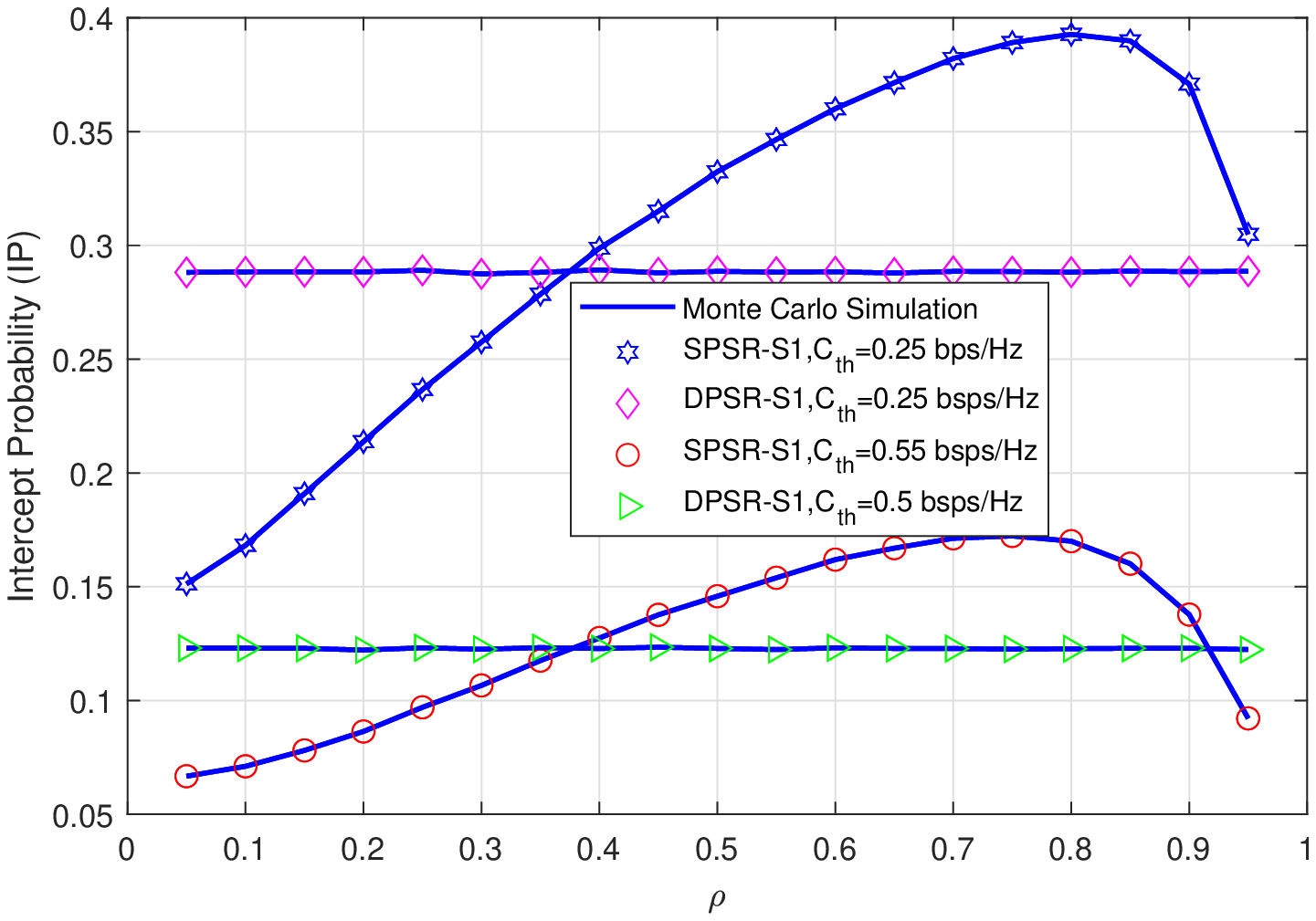}
		\caption{IP vs $\rho$ with $\eta$=0.8, $M$=2, $\Psi$=2 dB.}
		\label{fig:10}
	\end{minipage}
\vspace{-0.5cm}
\end{figure*}
\begin{figure*}[t]
	\centering
	\begin{minipage}{.48\textwidth}
		\begin{center}
			\includegraphics[width=9.5cm,height=8cm]{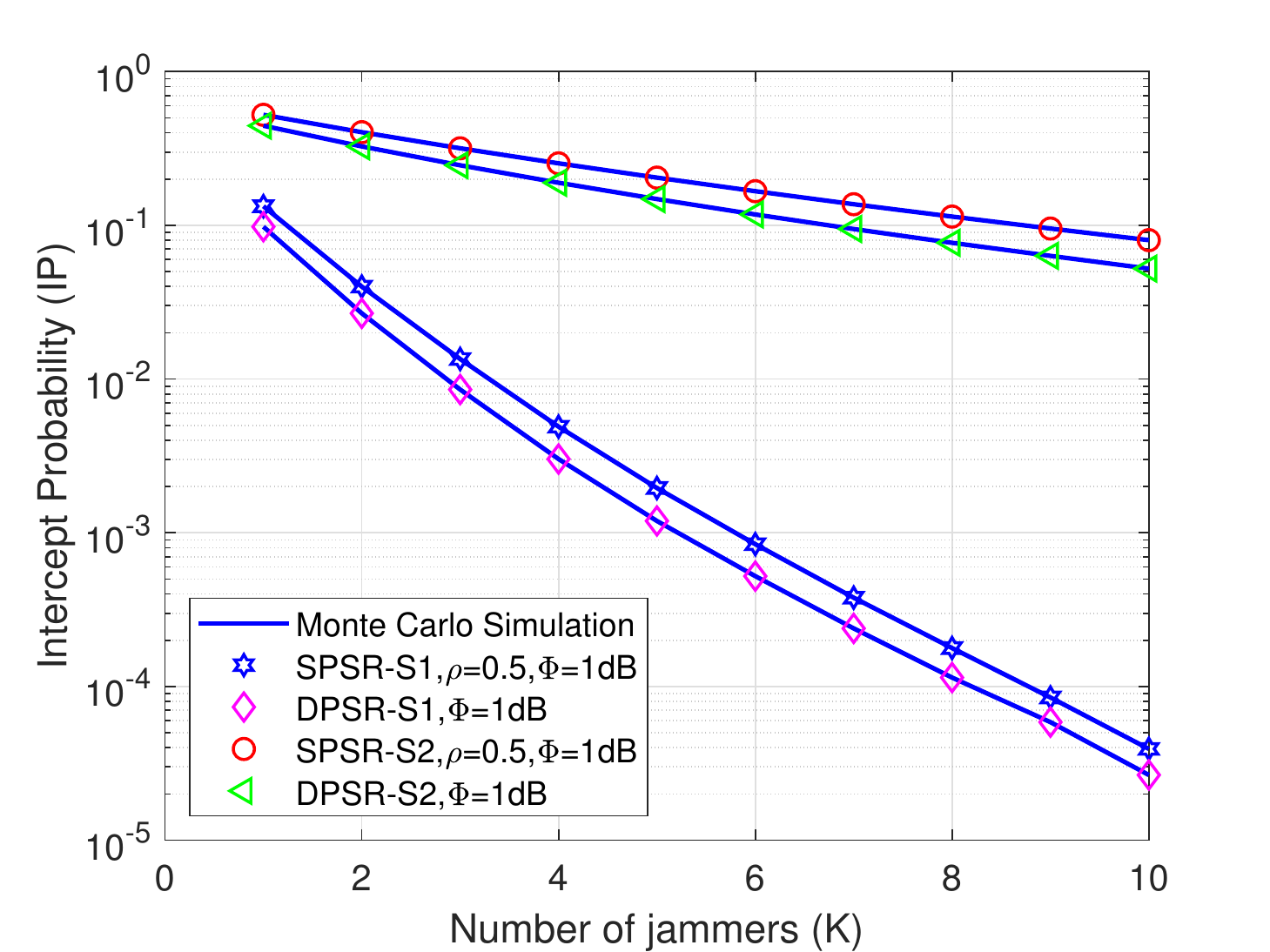} 
		\end{center}
		\caption{IP vs number of jammers (K) with $\eta$=0.8, $\rho$=0.5, $M$=3 and $\Psi$=2 dB. }
		\label{fig:11}
	\end{minipage} \hfill
	\begin{minipage}{.48\textwidth}
		\centering		
		\includegraphics[width=9.5cm,height=8cm]{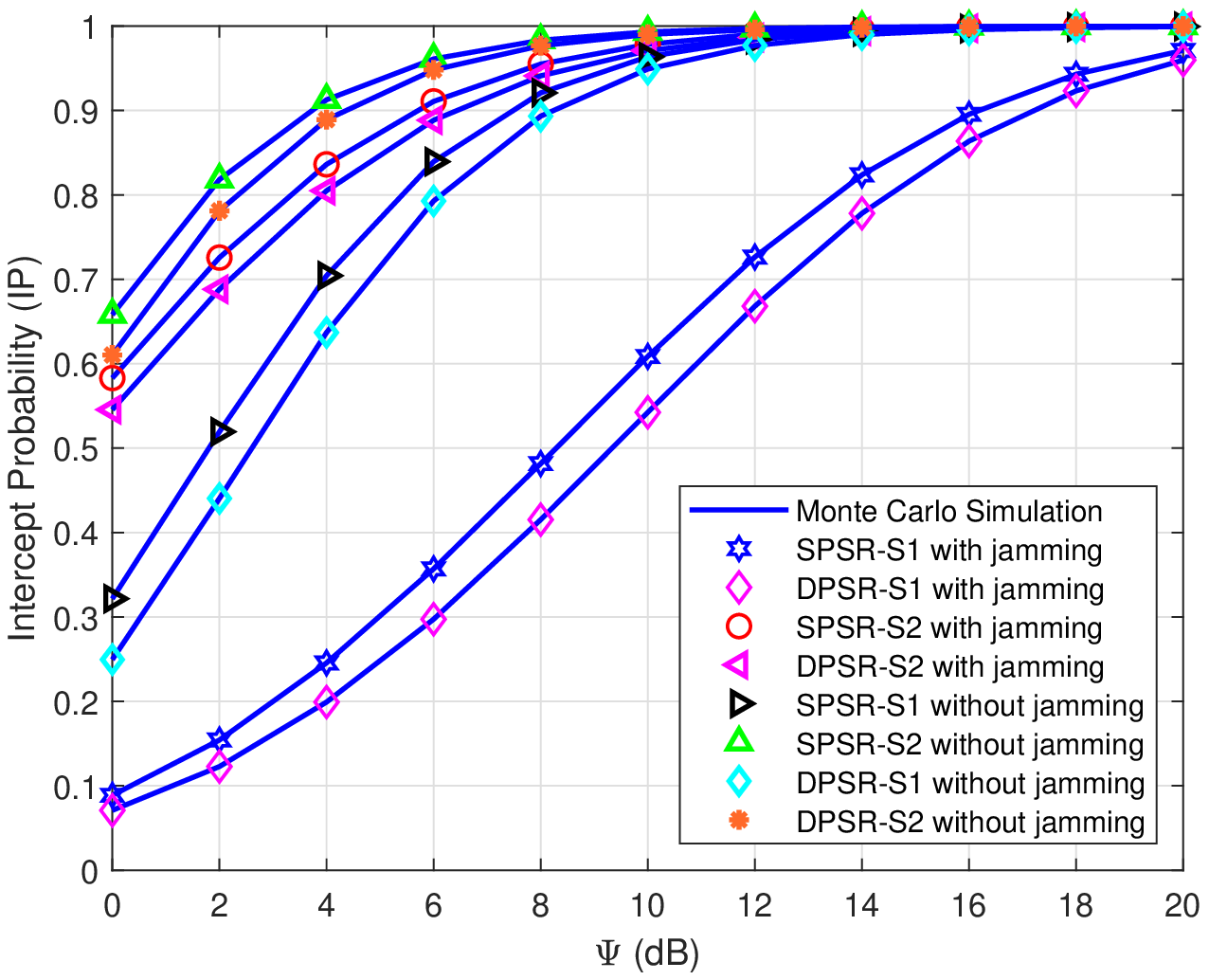}
		\caption{IP vs $\Psi$ with and without jamming cases.}
		\label{fig:12}
	\end{minipage}
\vspace{-0.5cm}
\end{figure*}

In Figs. \ref{fig:9} and \ref{fig:10}, we study OP and IP depending on different power splitting factor $\rho$, where $\eta$=0.8, $M$=2, $\Psi$=2 dB, $\Phi$=1 dB and K=1. The $\rho$ value plays an important role since it influences not only the amount of harvested energy at the relay but also the data transmission from ${\rm{R}} \to {\rm{D}}$. Therefore, there exists an optimal value of power splitting factor $\rho$ to maximize OP. Particularly, the OP of DPSR obtains the best performance compared to the two SPSR schemes since this method always optimizes the $\rho$ value during the system design. It also explains for the fact that why the DPSR curses are straight lines because DPSR does not depend on $\rho$ value. Besides, the DPSR scheme always obtains the best OP performance, and that does not mean it will have the worst IP performance, which is verified in Fig. \ref{fig:10}. This result is consistent with the result in Fig. \ref{fig:4}. For instance, when $0.35 < \rho  < 0.9$ the DPSR scheme has better a IP value than SPSR scheme. In order to provide an explicit explanation, we first observe that the IP is directly proportional to the power splitting factor as $0.35 < \rho  < 0.9$. In contrast, the DPSR yields the performance independent of the variation of $\rho$ since this benchmark exploits the optimal value $\rho^\ast$ as shown in Section~\ref{Sec:DPSRcase}. Furthermore, it holds that $\rho^\ast < \rho$ on average, and therefore the strength of received signal by the DPSR is better than the SPSR. Consequently, the DPSR provides lower the IP than the SPSR. Here again, the reliability-security trade-off phenomenon also happens in both figures.

To reduce the IP as well as the eavesdropping ability of E, we will consider the effect of the number of jammers as shown in Fig. \ref{fig:11}. In this figure, we can see that this will be consistent with reality because increasing K will cause more jamming signals to be sent to E and reduce its SNR ratio and its ability to overhear from the relay. Furthermore, in this figure, we compared DPSR with SPSR in IP analysis, while SPSR is considered in two cases with $\Phi$ equals -1 dB and 1 dB, respectively. It is easy to see that if the jammers' transmit power $\Phi$ is reduced, the jamming signal will be weaker and it is difficult to suppress the effect of the eavesdropping devices.

In Fig. 12, we investigate the IP as a function of $\Psi$ (dB) for the case with and without jamming, where $\eta=0.8$, $C_{\rm th}$ = 0.5 bps/Hz, M = 2, K = 1, $\rho$ = 0.55, and $\Phi$ = 1 dB. Firstly, it can be observed from Fig. 12 that the intercept performance of the SPSR schemes is much better than that as compared to the DPSR ones, which is has been explained in other figures. Furthermore, the intercept performance of the SPSR and DPSR with jamming is lower than that compared to without jamming scenarios. This is expected since the eavesdropper gains higher SNR in the case of no jamming signals. Thus, it has a higher probability to successfully decode the received signals from source and relay. Specifically, when $\Psi$ equals 2 and 4 dB, the IP of SPSR-S1 (or DPSR-S1) with jamming is 0.1542 (or 0.1228) and 0.2456 (or 0.1992), respectively. Meanwhile, the IP of the SPSR-S1 (or DPSR-S1) without jamming imposes 0.5194 (or 0.4404) and 0.7043 (or 0.6370), respectively. Moreover, the intercept performance of the S2 is better than that as compared to the S1. This is because the eavesdropper is located closer to sources in S2, thus it has a further distance to the jammer than S1 case. Consequently, the jammer has less effect on the eavesdropper in S2 than S1 schemes.

\section{{Conclusion and Future Directions}}
\label{Sec:Con}
In this paper, we have studied physical layer security for an AF-based SWIPT relay network consisting of multiple sources, multiple friendly jammers, an EH relay, and a destination in the presence of an eavesdropper. We have investigated the reliability-security trade-off performance of DPSR and SPSR schemes in terms of the IP and OP. Furthermore, the impact of system parameters on network performance, and the correctness of the analytical expressions have been verified and investigated by using Monte Carlo Simulations. Based on plots of the OP vs IP, we can recommend suitable system parameters to meet the pre-defined requirements on IP and OP. Particularly, a sufficient number of friendly jammers can significantly enhance system security by reducing the signal-to-noise ratio at the eavesdropper. The results of this paper can provide guidance for securing IoT networks in which the confidentiality in information transmission is important.  An interesting topic for further exploration in this area is the extension to cases in which source and eavesdropping nodes have multiple antennas.

The outcome of this work will motivate a more general model that considers a direct link between source and destination,  which imposes new challenges and complexities but might enhance network performance.

\appendix
\subsection{Proof of Theorem~\ref{Theorem:Outage1}} \label{Appendix:Outage1}
Based on the definition of the OP in \eqref{EQ21} with the SNR in \eqref{EQ11}, $\rm{OP}_{\rm{SPSR}}$ can be expressed as follows
\begin{align}
	\label{EQ22}
	&{\rm{O}}{{\rm{P}}_{{\rm{SPSR}}}} = \Pr \left( {\frac{{\eta \rho (1 - \rho )\Psi {\gamma _{{{\rm{S}}_m}{\rm{R}}}}{\gamma _{{\rm{RD}}}}}}{{\eta \rho {\gamma _{{\rm{RD}}}} + (1 - \rho )}} < {\gamma _{th}}} \right) \notag\\ 
	&= \Pr \left( {{\gamma _{{{\rm{S}}_m}{\rm{R}}}} < \frac{{{\gamma _{th}}\left( {\eta \rho {\gamma _{{\rm{RD}}}} + (1 - \rho )} \right)}}{{\eta \rho (1 - \rho )\Psi {\gamma _{{\rm{RD}}}}}}} \right)\notag\\
	&= \int\limits_0^\infty  {{F_{{\gamma _{{{\rm{S}}_m}{\rm{R}}}}}}\left\{ {\frac{{{\gamma _{th}}\left( {\eta \rho x + (1 - \rho )} \right)}}{{\eta \rho (1 - \rho )\Psi x}}} \right\}{f_{{\gamma _{{\rm{RD}}}}}}(x)} dx,
\end{align}
where the CDF of ${\gamma _{{{\rm{S}}_m}{\rm{R}}}}$ is given in \eqref{EQ18} and the PDF of  ${f_{{\gamma _{{\rm{RD}}}}}}(x)$ is given in \eqref{EQ2}. By utilizing \eqref{EQ18}, we obtain the following expression
\begin{multline} \label{eq:FsmR}
{F_{{\gamma _{{{\rm{S}}_m}{\rm{R}}}}}}\left\{ {\frac{{{\gamma _{th}}\left( {\eta \rho x + \left( {1 - \rho } \right)} \right)}}{{\eta \rho \left( {1 - \rho } \right)\Psi x}}} \right\} = \\
 1 + \sum\limits_{b = 1}^M {{{\left( { - 1} \right)}^b}C_M^b\exp \left( { - \frac{{b{\lambda _{{\rm{SR}}}}{\gamma _{th}}}}{{\eta \rho \Psi x}} - \frac{{b{\lambda _{{\rm{SR}}}}{\gamma _{th}}}}{{\left( {1 - \rho } \right)\Psi }}} \right)}.
\end{multline}
After that, by plugging \eqref{EQ2} and \eqref{eq:FsmR} into \eqref{EQ22}, $\rm{OP}_{\rm{SPSR}}$ can be represented as 
\begin{align}
	\label{EQ23}
	{\rm{O}}{{\rm{P}}_{{\rm{SPSR}}}} &= 1 + \sum\limits_{b = 1}^M {{{( - 1)}^b}C_M^b} \exp \left( { - \frac{{b{\lambda _{{\rm{SR}}}}{\lambda _{th}}}}{{(1 - \rho )\Psi }}} \right) \notag\\ & \times {\lambda _{{\rm{RD}}}}\int\limits_0^\infty  {\exp } \left( { - \frac{{b{\lambda _{{\rm{SR}}}}{\gamma _{th}}}}{{\eta \rho \Psi x}} - {\lambda _{{\rm{RD}}}}x} \right)dx.
\end{align}
With the help of \cite[3.324.1]{gradshteyn2014} and some algebraic manipulations, we obtain \eqref{eq:Theorem_1}, which completes the proof.
\subsection{Proof of Theorem~\ref{Theorem:IP}} \label{Appendix:IP}
Based on \eqref{eq:gammaE1}, \eqref{EQ15}, \eqref{EQ16}, and \eqref{EQ25},  the ${\rm{I}}{{\rm{P}}_{{\rm{SPSR}}}}$ can be rewritten by:
\begin{equation} \label{EQ26}
\begin{split}
	&{\text{I}}{{\text{P}}_{{\text{SPSR}}}} = 1 - \Pr \left( { \alpha < {\gamma _{th}}} \right)  \\
&	= 1 - \int\limits_0^\infty  {\Pr \left( {\max \left( {{\vartheta _1},{\vartheta _2}} \right) < {\gamma _{th}}} \right) {f_\Xi }(x)dx}  \hfill \\
&	= 1 - \int\limits_0^\infty  {Q(x) {f_\Xi }(x)dx} , \hfill \\ 
\end{split} 
	\end{equation}
	where $\alpha = \max \left( {\frac{{\Psi {\gamma _{{{\text{S}}_m}{\text{E}}}}}}{{\Phi \Xi }},\frac{{\eta \rho \left( {1 - \rho } \right)\Psi {\gamma _{{{\text{S}}_m}{\text{R}}}}{\gamma _{{\text{RE}}}}}}{{\eta \rho {\gamma _{{\text{RE}}}} + \left( {1 - \rho } \right)\Phi \Xi  + \left( {1 - \rho } \right)}}} \right)$ and the following definitions hold
	\begin{align}
		\vartheta_1 &\triangleq {\frac{{\Psi {\gamma _{{{\rm{S}}_m}{\rm{E}}}}}}{{\Phi x }},}\\
		\vartheta_2 &\triangleq {\frac{{\eta \rho (1 - \rho ){\gamma _{{{\rm{S}}_m}{\rm{R}}}}{\gamma _{{\rm{RE}}}}\Psi }}{{\eta \rho {\gamma _{{\rm{RE}}}} + \left( {1 - \rho } \right)\Phi x  + \left( {1 - \rho } \right)}} }, \\
		Q(x) &\triangleq \Pr \left( {\max \left(\vartheta_1, \vartheta_2 \right)  < {\gamma _{th}}} \right) \notag\\
		&= \underbrace {\Pr \left( 	\vartheta_1 < {\gamma _{th}} \right)}_{{Q_1}(x)} \times \underbrace {\Pr \left( 	\vartheta_2 < {\gamma _{th}} \right)}_{{Q_2}(x)}  \label{EQQv1}
\end{align}
From \eqref{EQQv1}, $Q_1(x)$ and $Q_2(x)$  can be calculated by:
	\begin{align}
		\label{EQ27}
		&{Q_1}(x) = \Pr \left( {\frac{{\Psi {\gamma _{{{\text{S}}_m}{\text{E}}}}}}{{\Phi x}} < {\gamma _{th}}} \right) = \Pr \left( {{\gamma _{{{\text{S}}_m}{\text{E}}}} < \frac{{{\gamma _{th}}\Phi x}}{\Psi }} \right) \hfill \notag\\
		&= 1 - \exp \left( { - \frac{{{\gamma _{th}}{\lambda _{{\text{SE}}}}\Phi x}}{\Psi }} \right), \\ 
		\label{EQQ2}
		& {Q_2}(x) = \Pr \left( {\frac{{\eta \rho {\rho _1}\Psi {\gamma _{{{\text{S}}_m}{\text{R}}}}{\gamma _{{\text{RE}}}}}}{{\eta \rho {\gamma _{{\text{RE}}}} + {\rho _1}\Phi x + {\rho _1}}} < {\gamma _{th}}} \right)  \notag \\
			& = \Pr \left\{ {{\gamma _{{{\text{S}}_m}{\text{R}}}} < \frac{{{\gamma _{th}}\left( {\eta \rho {\gamma _{{\text{RE}}}} + {\rho _1}\Phi x + {\rho _1}} \right)}}{{\eta \rho {\rho _1}{\gamma _{{\text{RE}}}}\Psi }}} \right\}  \notag \\
			& = \int\limits_0^\infty  {{F_{{\gamma _{{{\text{S}}_m}{\text{R}}}}}}} \left\{ {\frac{{{\gamma _{th}}\left( {\eta \rho y + {\rho _1}\Phi x + {\rho _1}} \right)}}{{\eta \rho {\rho _1}y\Psi }}} \right\}  {f_{{\gamma _{{\text{RE}}}}}}(y)dy, 
	\end{align}
	where $\rho_1 \triangleq (1-\rho)$. By using the same approach as what has done in \eqref{eq:FsmR} to the last equation of \eqref{EQQ2}, we obtain the closed-form expression of $Q_2(x)$
	\begin{align}
		\label{EQ28}
		&{Q_2(x)}  = 1 + 2\sum\limits_{b = 1}^M {{{( - 1)}^b}} C_M^b\exp \left( { - \frac{{b{\lambda _{{\rm{SR}}}}{\gamma _{th}}}}{{\rho_1 \Psi }}} \right)\times\notag\\
		& \sqrt {\frac{{b{\lambda _{{\rm{SR}}}}{\lambda _{{\rm{RE}}}}{\gamma _{th}}\left( {\Phi x + 1} \right)}}{{\eta \rho \Psi }}}  {K_1}\left( {2\sqrt {\frac{{b{\lambda _{{\rm{SR}}}}{\lambda _{{\rm{RE}}}}{\gamma _{th}}\left( {\Phi x + 1} \right)}}{{\eta \rho \Psi }}} } \right).
\end{align} 
Following the same methodology as done to obtain the result in \eqref{eq:Theorem_1}, we obtain the closed-form expression of $Q_2(x)$ as
\begin{equation} \label{EQ28v1}
	\begin{split}
		& {Q_2}(x) = 1 + 2\sum\limits_{b = 1}^M {{{\left( { - 1} \right)}^b}C_M^b\exp \left( { - \frac{{b{\lambda _{{\text{SR}}}}{\gamma _{th}}}}{{{\rho _1}\Psi }}} \right)}   \\
		&  \times \sqrt {\frac{{b{\lambda _{{\text{SR}}}}{\gamma _{{\text{RE}}}}{\gamma _{th}}\left( {\Phi x + 1} \right)}}{{\eta \rho \Psi }}}  \times {K_1}\left( {2\sqrt {\frac{{b{\lambda _{{\text{SR}}}}{\gamma _{{\text{RE}}}}{\gamma _{th}}\left( {\Phi x + 1} \right)}}{{\eta \rho \Psi }}} } \right).
	\end{split}
\end{equation}
Based on \eqref{EQ20}, \eqref{EQ27}, and \eqref{EQ28v1}, the ${\rm{I}}{{\rm{P}}_{{\rm{SPSR}}}}$ in \eqref{EQ26} can be reformulated as in \eqref{EQ30},
\begin{figure*}
\begin{align}
		\label{EQ30}
		& {\rm{I}}{{\rm{P}}_{{\rm{SPSR}}}} = 1 - \int\limits_0^\infty  {\left\{ {1 - \exp \left( { - \frac{{{\gamma _{th}}{\lambda _{{\rm{SE}}}}\Phi x}}{\Psi }} \right)} \right\}}  \left\{ \begin{array}{l}
			1 + 2\sum\limits_{b = 1}^M {{{\left( { - 1} \right)}^b}} C_M^b\exp \left( { - \frac{{b{\lambda _{{\rm{SR}}}}{\gamma _{th}}}}{{{\rho _1}\Psi }}} \right)\\
			\times \sqrt {\frac{{b{\lambda _{{\rm{SR}}}}{\lambda _{{\rm{RE}}}}{\gamma _{th}}\left( {\Phi x + 1} \right)}}{{\eta \rho \Psi }}} \\
			\times {K_1}\left( {2\sqrt {\frac{{b{\lambda _{{\rm{SR}}}}{\lambda _{{\rm{RE}}}}{\gamma _{th}}\left( {\Phi x + 1} \right)}}{{\eta \rho \Psi }}} } \right)
		\end{array} \right\}  \frac{{{{\left( {{\lambda _{{\rm{JE}}}}} \right)}^K}}}{{\left( {K - 1} \right)!}}{x^{K - 1}}\exp \left( { - {\lambda _{{\rm{JE}}}}x} \right)dx  \notag\\
		&= \underbrace {\frac{{{{\left( {{\lambda _{{\rm{JE}}}}} \right)}^K}}}{{\left( {K - 1} \right)!}}\int\limits_0^\infty  {{x^{K - 1}}\exp \left( { - x{{\tilde \lambda }_{{\rm{JE}}}}} \right)dx} }_{{\Upsilon _1}}+ 2\sum\limits_{b = 1}^M {{{\left( { - 1} \right)}^b}} C_M^b \exp \left( { - \frac{{b{\lambda _{{\rm{SR}}}}{\gamma _{th}}}}{{{\rho _1}\Psi }}} \right) \notag\\
		&\times \frac{{{{\left( {{\lambda _{{\rm{JE}}}}} \right)}^K}}}{{\left( {K - 1} \right)!}}  \sqrt {\frac{{b{\lambda _{{\rm{SR}}}}{\lambda _{{\rm{RE}}}}{\gamma _{th}}}}{{\eta \rho \Psi }}} \left\{ \begin{array}{l}
			\underbrace {\int\limits_0^\infty  \begin{array}{l}
					{x^{K - 1}}\exp \left( { - x{{\tilde \lambda }_{{\rm{JE}}}}} \right) \times \sqrt {\left( {\Phi x + 1} \right)} \\
					\times {K_1}\left( {2\sqrt {\frac{{b{\lambda _{{\rm{SR}}}}{\lambda _{{\rm{RE}}}}{\gamma _{th}}\left( {\Phi x + 1} \right)}}{{\eta \rho \Psi }}} } \right)dx
			\end{array} }_{{\Upsilon _2}}\\
			- \underbrace {\int\limits_0^\infty  \begin{array}{l}
					{x^{K - 1}}\exp \left( { - {\lambda _{{\rm{JE}}}}x} \right)\sqrt {\left( {\Phi x + 1} \right)} \\
					\times {K_1}\left( {2\sqrt {\frac{{b{\lambda _{{\rm{SR}}}}{\lambda _{{\rm{RE}}}}{\gamma _{th}}\left( {\Phi x + 1} \right)}}{{\eta \rho \Psi }}} } \right)dx
			\end{array} }_{{\Upsilon _3}}
		\end{array} \right\},
	\end{align}
\hrule
\end{figure*}
where ${\tilde \lambda _{{\rm{JE}}}} = \frac{{{\gamma _{th}}{\lambda _{{\rm{SE}}}}\Phi }}{\Psi } + {\lambda _{{\rm{JE}}}}$. Thanks to \cite[3.381.4]{gradshteyn2014}, ${\Upsilon _1}$ in \eqref{EQ30} can be calculated as:
	\begin{align}
		\label{EQUpsilon1}
		{\Upsilon _1} = {\left( {\frac{{{\lambda _{{\rm{JE}}}}}}{{{{\tilde \lambda }_{{\rm{JE}}}}}}} \right)^K}.
\end{align}
By applying Taylor series for $\exp \left( { - {\lambda _{{\rm{JE}}}}x} \right) = \sum\limits_{t = 0}^\infty  {\frac{{{{\left( { - {\tilde  \lambda _{{\rm{JE}}}}x} \right)}^t}}}{{t!}}}  = \sum\limits_{t = 0}^\infty  {\frac{{{{( - 1)}^t}{{\left( {{\tilde  \lambda _{{\rm{JE}}}}} \right)}^t}}}{{t!}}} {x^t}$ and by changing the variable $y = \Phi x + 1$, $\Upsilon _2$ can be rewritten as:
	\begin{align}
		\label{EQ31}
		{\Upsilon _2} &= \sum\limits_{t = 0}^\infty  {\frac{{{{( - 1)}^t}{{\left( {{{\tilde \lambda }_{{\rm{JE}}}}} \right)}^t}}}{{t!{\Phi ^{t + K}}}}} \int\limits_1^\infty  {y^{1/2}}{{\left( {y - 1} \right)}^{t + K - 1}} \notag\\ &\times {K_1} \left( {2\sqrt {\frac{{b{\lambda _{{\rm{SR}}}}{\lambda _{{\rm{RE}}}}{\gamma _{th}}y}}{{\eta \rho \Psi }}} } \right)dy.
\end{align}
Next, by using \cite[6.592.4]{gradshteyn2014}, above equation can be reformulated by:
	\vspace{0.2cm}
	\begin{align}
		\label{EQ32}
		{\Upsilon _2}=& \sum\limits_{t = 0}^\infty  {\frac{{{{( - 1)}^t}{{\left( {{{\tilde \lambda }_{{\rm{JE}}}}} \right)}^t}\Gamma \left( {t + K} \right)}}{{t!{\Phi ^{t + K}}}}}  \times {\left( {2\sqrt {\frac{{b{\lambda _{{\rm{SR}}}}{\lambda _{{\rm{RE}}}}{\gamma _{th}}}}{{\eta \rho \Psi }}} } \right)^{ - 1}} \notag\\ & \times G_{1,3}^{3,0}\left( {\frac{{b{\lambda _{{\rm{SR}}}}{\lambda _{{\rm{RE}}}}{\gamma _{th}}}}{{\eta \rho \Psi }}\left| \begin{array}{l}
				0\\
				- t - K,1,0
			\end{array} \right.} \right),
	\end{align}
	where $\Gamma \left(  \cdot  \right)$ is the Gamma function and  $G_{p,q}^{m,n}\left( {z\left| \begin{array}{l}
			{a_1},...,{a_p}\\
			{b_1},...,{b_q}
		\end{array} \right.} \right)$ is the Meijer G-function. By applying the same approach for \eqref{EQ32},  $\Upsilon _3$ can be formulated as:
	\begin{align}
		\label{EQUpsilon3}
		{\Upsilon _3} &= \sum\limits_{t = 0}^\infty  {\frac{{{{( - 1)}^t}{{\left( {{\lambda _{{\rm{JE}}}}} \right)}^t}\Gamma \left( {t + K} \right)}}{{t!{\Phi ^{t + K}}}}}  \times {\left( {2\sqrt {\frac{{b{\lambda _{{\rm{SR}}}}{\lambda _{{\rm{RE}}}}{\gamma _{th}}}}{{\eta \rho \Psi }}} } \right)^{ - 1}} \notag\\ &\times G_{1,3}^{3,0}\left( {\frac{{b{\lambda _{{\rm{SR}}}}{\lambda _{{\rm{RE}}}}{\gamma _{th}}}}{{\eta \rho \Psi }}\left| \begin{array}{l}
				0\\
				- t - K,1,0
			\end{array} \right.} \right).
\end{align}
Finally, By substituting \eqref{EQUpsilon1}, \eqref{EQ32}, and \eqref{EQUpsilon3}, the ${\rm{I}}{{\rm{P}}_{{\rm{SPSR}}}}$, we obtain \eqref{eq:Theorem_2}. This is the end of the proof.
\subsection{Proof of Theorem~\ref{Theorem:Outage2}} \label{Appendix:Outage2}
Substituting the optimal ${\rho ^*} = \frac{1}{{1 + \left| {{h_{{\rm{RD}}}}} \right|\sqrt \eta  }}$ into \eqref{EQ22}, the ${\rm{O}}{{\rm{P}}_{{\rm{DPSR}}}}$ is expressed as
\begin{align}
	\label{EQ34}
	&{\rm{O}}{{\rm{P}}_{{\rm{DPSR}}}} \notag\\ &= \Pr \left( {\frac{{\eta \Psi {\gamma _{{{\rm{S}}_m}{\rm{R}}}}{\gamma _{{\rm{RD}}}}}}{{{{\left( {1 + \sqrt {\eta {\gamma _{{\rm{RD}}}}} } \right)}^2}}} < {\gamma _{th}}} \right) \notag\\ &= \Pr \left( {{\gamma _{{{\rm{S}}_m}{\rm{R}}}} < \frac{{{\gamma _{th}}\left( {1 + 2\sqrt {\eta {\gamma _{{\rm{RD}}}}}  + \eta {\gamma _{{\rm{RD}}}}} \right)}}{{\eta \Psi {\gamma _{{\rm{RD}}}}}}} \right)\notag\\
	&= \int\limits_0^\infty  {{F_{{\gamma _{{{\rm{S}}_m}{\rm{R}}}}}}\left( {\frac{{{\gamma _{th}}\left( {1 + 2\sqrt {\eta x}  + \eta x} \right)}}{{\eta \Psi x}}} \right) \times {f_{{\gamma _{{\rm{RD}}}}}}(x)dx} \notag\\
	&= 1 + \sum\limits_{b = 1}^M {{{( - 1)}^b}C_M^b} \exp \left( { - \frac{{b{\lambda _{{\rm{SR}}}}{\gamma _{th}}}}{\Psi }} \right) \notag\\ &\int\limits_0^\infty  {{\lambda _{{\rm{RD}}}}\exp \left( { - \frac{{2b{\lambda _{{\rm{SR}}}}{\gamma _{th}}}}{{\Psi \sqrt {\eta x} }}} \right)\exp \left( { - \frac{{b{\lambda _{{\rm{SR}}}}{\gamma _{th}}}}{{\eta \Psi x}} - {\lambda _{{\rm{RD}}}}x} \right)dx} .
\end{align}
By adopting Taylor series for $\exp \left( { - \frac{{2b{\lambda _{{\rm{SR}}}}{\gamma _{th}}}}{{\Psi \sqrt {\eta x} }}} \right) = \sum\limits_{t = 0}^\infty  {\frac{{{{( - 1)}^t}{2^t}}}{{t!}}} {\left( {\frac{{b{\lambda _{{\rm{SR}}}}{\gamma _{th}}}}{{\Psi \sqrt \eta  }}} \right)^t}{x^{ - t/2}}$, equation \eqref{EQ34} can be rewritten as follows
\begin{align}
	\label{EQ35}
	&{\rm{O}}{{\rm{P}}_{{\rm{DPSR}}}} \notag\\&= 1 + \sum\limits_{t = 0}^\infty  {\sum\limits_{b = 1}^M {\frac{{{{( - 1)}^{t + b}}C_M^b{2^t}{\lambda _{{\rm{RD}}}}}}{{t!}}{{\left( {\frac{{b{\lambda _{{\rm{SR}}}}{\gamma _{th}}}}{{\Psi \sqrt \eta  }}} \right)}^t}} } \notag\\
	& \exp \left( { - \frac{{b{\lambda _{{\rm{SR}}}}{\gamma _{th}}}}{\Psi }} \right) \int\limits_0^\infty  {{x^{ - t/2}}\exp \left( { - \frac{{b{\lambda _{{\rm{SR}}}}{\gamma _{th}}}}{{\eta \Psi x}} - {\lambda _{{\rm{RD}}}}x} \right)dx}.
\end{align}
Finally, by applying \cite[3.471.9]{gradshteyn2014}, we obtain \eqref{eq:Theorem_3}, which finises the proof.
\subsection{Proof of Theorem~\ref{Theorem:IPv2}} \label{Appendix:IPv2}
By substituting ${\rho ^*} = \frac{1}{{1 + \left| {{h_{{\rm{RD}}}}} \right|\sqrt \eta  }}$ into \eqref{EQ26}, $Q_2^*$ can be calculated as
	\begin{align}
		\label{EQ37}
		&Q_2^*(w)= \Pr \left( {\frac{{\eta {\gamma _{{{\rm{S}}_m}{\rm{R}}}}{\gamma _{{\rm{RE}}}}\Psi }}{{\eta {\Phi _1}{\gamma _{{\rm{RE}}}} + {\Phi _2}\Phi x + {\Phi _2}}} < {\gamma _{th}}} \right) \notag\\
		&= \int\limits_0^\infty  {\underbrace {\Pr \left( {\frac{{\eta {\gamma _{{{\rm{S}}_m}{\rm{R}}}}{\gamma _{{\rm{RE}}}}\Psi }}{{\eta {\Phi _3}{\gamma _{{\rm{RE}}}} + {\Phi _4}\Phi x + {\Phi _4}}} < {\gamma _{th}}} \right)}_\Theta {f_{{\gamma _{{\rm{RD}}}}}}\left( \omega  \right)d\omega } ,
\end{align}
where $\Phi_1 \triangleq \left( {\frac{{1 + \sqrt {{\gamma _{{\rm{RD}}}}\eta } }}{{\sqrt {{\gamma _{{\rm{RD}}}}\eta } }}} \right)$, $\Phi_2 \triangleq  {1 + \sqrt {{\gamma _{{\rm{RD}}}}\eta } } $, $\Phi_3 \triangleq \left( {\frac{{1 + \sqrt {\eta w} }}{{\sqrt {\eta w} }}} \right)$, $\Phi_4 \triangleq  {1 + \sqrt {w\eta } } $.
Based on \eqref{EQ19} and then with the help of \cite[3.324.1]{gradshteyn2014}, $\Theta $ can be computed as:
	\begin{align}
		\label{EQ39}
		&\Theta( x)  = \Pr \left( {\frac{{\eta {\gamma _{{{\rm{S}}_m}{\rm{R}}}}{\gamma _{{\rm{RE}}}}\Psi }}{{\eta \Phi_3 {\gamma _{{\rm{RE}}}} + \Phi_4 \Phi x + \Phi_4}} < {\gamma _{th}}} \right) \notag\\
		&= \int\limits_0^\infty  {\Pr \left( {{\gamma _{{{\rm{S}}_m}{\rm{R}}}} < \frac{{{\gamma _{th}}\left( {\eta \Phi_3 z + \Phi_4\Phi x + \Phi_4 } \right)}}{{\eta z\Psi }}} \right) {f_{{\gamma _{{\rm{RE}}}}}}(z)dz} \notag\\
		&= 1 + \sum\limits_{b = 1}^M {{{\left( { - 1} \right)}^b}} C_M^b\exp \left( { - \frac{{b{\lambda _{{\rm{SR}}}}{\gamma _{th}}\Phi_3 }}{\Psi }} \right) \notag\\
		&= 1 + \sum\limits_{b = 1}^M {{{\left( { - 1} \right)}^b}} C_M^b\exp \left( { - \frac{{b{\lambda _{{\rm{SR}}}}{\gamma _{th}}\Phi_3 }}{\Psi }} \right) \\
		&\times \int\limits_0^\infty  {{\lambda _{{\rm{RE}}}}\exp \left( { - \frac{{b{\lambda _{{\rm{SR}}}}{\gamma _{th}}\left( {\Phi x + 1} \right)\Phi_4}}{{\eta z\Psi }} - {\lambda _{{\rm{RE}}}}z} \right)} dz \notag \\
		&= 1 + 2\sum\limits_{b = 1}^M {{{\left( { - 1} \right)}^b}} C_M^b\exp \left( { - \frac{{b{\lambda _{{\rm{SR}}}}{\gamma _{th}}\Phi_3 }}{\Psi }} \right) \notag\\
		&\times \sqrt {\frac{{b{\lambda _{{\rm{SR}}}}{\lambda _{{\rm{RE}}}}{\gamma _{th}}\left( {\Phi x + 1} \right)\Phi_4 }}{{\eta \Psi }}} \notag\\ & \times {K_1}\left( {2\sqrt {\frac{{b{\lambda _{{\rm{SR}}}}{\lambda _{{\rm{RE}}}}{\gamma _{th}}\left( {\Phi x + 1} \right)\Phi_4 }}{{\eta \Psi }}} } \right).
\end{align}
By substituting \eqref{EQ39} into \eqref{EQ20}, the term $Q_2^*$ can be reformulated as follows
	\begin{align}
		\label{EQ40}
		Q_2^*(w) &= 1 + 2\sum\limits_{b = 1}^M {{{\left( { - 1} \right)}^b}C_M^b} \exp \left( { - \frac{{b{\lambda _{{\rm{SR}}}}{\gamma _{th}}}}{\Psi }} \right) \notag\\
		&\times \sqrt {\frac{{b{\lambda _{{\rm{SR}}}}{\lambda _{{\rm{RE}}}}{\gamma _{th}}\left( {\Phi x + 1} \right)}}{{\eta \Psi }}} \notag\\
		&\times {\lambda _{{\rm{RD}}}}\int\limits_0^\infty  \begin{array}{l}
			\sqrt {1 + \sqrt {\eta \omega } } \exp \left( { - \frac{{b{\lambda _{{\rm{SR}}}}{\gamma _{th}}}}{{\Psi \sqrt {\eta \omega } }} - {\lambda _{{\rm{RD}}}}\omega } \right)\\
			{K_1}\left( {2\sqrt {\frac{{b{\lambda _{{\rm{SR}}}}{\lambda _{{\rm{RE}}}}{\gamma _{th}}\left( {\Phi x + 1} \right)\left( {1 + \sqrt {\eta \omega } } \right)}}{{\eta \Psi }}} } \right)d\omega 
		\end{array}
	\end{align}
\begin{align}
		&= 1 + 2\sum\limits_{b = 1}^M {{{\left( { - 1} \right)}^b}C_M^b} \exp \left( { - \frac{{b{\lambda _{{\rm{SR}}}}{\gamma _{th}}}}{\Psi }} \right) \notag\\
		&\times \sqrt {\frac{{b{\lambda _{{\rm{SR}}}}{\lambda _{{\rm{RE}}}}{\gamma _{th}}\left( {\Phi x + 1} \right)}}{{\eta \Psi }}}  \times \Delta (\omega ),
\end{align}
	where the following definition holds
	\begin{equation}
	 \begin{array}{l}
		\Delta (\omega ) \buildrel \Delta \over = {\lambda _{{\rm{RD}}}}\int\limits_0^\infty  {\sqrt {1 + \sqrt {\eta \omega } } \exp \left( { - \frac{{b{\lambda _{{\rm{SR}}}}{\gamma _{th}}}}{{\Psi \sqrt {\eta \omega } }} - {\lambda _{{\rm{RD}}}}\omega } \right)} \\
		\times {K_1}\left( {2\sqrt {\frac{{b{\lambda _{{\rm{SR}}}}{\lambda _{{\rm{RE}}}}{\gamma _{th}}\left( {\Phi x + 1} \right)\left( {1 + \sqrt {\eta \omega } } \right)}}{{\eta \Psi }}} } \right)d\omega .
	\end{array}
\end{equation}
Finally, based on \eqref{EQ26}, \eqref{EQ27}, and \eqref{EQ40}, we claim \eqref{eq:Theorem_4}, which completes the proof.


\bibliographystyle{IEEEtran}
\bibliography{IEEEfull}

\begin{IEEEbiographynophoto} 
	{Tan N. Nguyen} was born at Nha Trang City, Vietnam, in 1986. He received B.S. and M.S. degrees in electronics and telecommunications engineering from Ho Chi Minh University of Natural Sciences, a member of Vietnam National University at Ho Chi Minh City (Vietnam) in 2008 and 2012, respectively. He is currently pursuing his Ph.D. degree in electrical engineering at VSB Technical University of Ostrava, Czech Republic. He got his Ph.D. degree in computer science, communication technology and applied mathematics at VSB Technical University of Ostrava, Czech Republic, in 2019. In 2013, he joined the Faculty of Electrical and Electronics Engineering of Ton Duc Thang University, Vietnam and have been working as lecturer since then. His major interests are cooperative communications, cognitive radio, and physical layer security.
\end{IEEEbiographynophoto}
\begin{IEEEbiographynophoto} 
	{Dinh-Hieu TRAN} (S'20) was born and grew up in Gia Lai, Vietnam (1989). He received a B.E. degree in Electronics and Telecommunication Engineering Department from Ho Chi Minh City University of Technology, Vietnam, in 2012. He finished his M.Sc degree in Electronics and Computer Engineering from Hongik University, Korea, in 2017, and the Ph.D. degree in Telecommunications Engineering from the Interdiscipline Reliability and Trust (SnT) research center, the University of Luxembourg in December 2021, under the supervision of Prof. Symeon Chatzinotas and Prof. Björn Ottersten. His major interests include UAV, Satellite, IoT, Mobile Edge Computing, Caching, B5G for wireless communication networks. In 2016, he received the Hongik Rector Award for his excellence during his master's study at Hongik University. He was a co-recipient of the IS3C 2016 best paper award. In 2021, he was nominated for the Best Ph.D. Thesis Award at the University of Luxembourg.
\end{IEEEbiographynophoto}
\begin{IEEEbiographynophoto} 
	{Trinh Van Chien} (S'16-M'20) received the B.S. degree in Electronics and Telecommunications from Hanoi University of Science and Technology (HUST), Vietnam, in 2012. He then received the M.S. degree in Electrical and Computer Enginneering from Sungkyunkwan University (SKKU), Korea, in 2014 and the Ph.D. degree in Communication Systems from Link\"oping University (LiU), Sweden, in 2020. He was  a research associate at University of Luxembourg. He is now with the School of Information and Communication Technology (SoICT), Hanoi University of Science and Technology (HUST), Vietnam. His interest lies in convex optimization problems and machine learning applications for wireless communications and image \& video processing. He was an IEEE wireless communications letters exemplary reviewer for 2016, 2017 and 2021. He also received the award of scientific excellence in the first year of the 5Gwireless project funded by European Union Horizon's 2020.
\end{IEEEbiographynophoto}
\begin{IEEEbiographynophoto} 
	{Van-Duc Phan} was born in 1975 in  Long An province, Vietnam. He received his M.S. degree in Department of Electric, Electrical and Telecommunication Engineering from Ho Chi Minh City University of Transport, Ho Chi Minh, Vietnam and Ph.D. degree in Department of Mechanical and Automation Engineering, Da-Yeh University, Taiwan in 2016. Currently, he is a research interests are in sliding mode control, non-linear systems or active manegtic bearing, flywheel store energy systems, power system optimization, optimization algorithms, and renewable energies, Energy harvesting (EH) enabled cooperative networks, Improving the optical properties, lighting performance of white LEDs, Energy efficiency LED driver integrated circuits, Novel radio access technologies, Physical security in communication network
\end{IEEEbiographynophoto}
\begin{IEEEbiographynophoto} 
	{Miroslav Voznak} (Senior Member, IEEE)
	was born in 1971. He received the Ph.D. degree
	in telecommunications from the VSB, Technical
	University of Ostrava, Czech Republic, in 2002.
	He is currently a Professor of electronics and com-
	munication technologies with the Department of
	Telecommunications, VSB, Technical University
	of Ostrava, and a Foreign Professor with Ton Duc
	Thang University, Ho Chi Minh City, Vietnam.
	He is a coauthor more than one hundred articles
	in journals indexed in the SCIE database. His research interests include IP
	telephony, wireless networks, network security, and big data analytics.
\end{IEEEbiographynophoto}
\begin{IEEEbiographynophoto} 
	{Phu Tran Tin} was born in Khanh Hoa, Viet Nam, in 1979. He received the Bachelor's degree (2002) and Master's degree (2008) from Ho Chi Minh City University of Science. He is currently a lecturer at the Faculty of Electronics Technology (FET), Industrial University of Ho Chi Minh City. In 2019, he received the Ph.D. degree in Faculty of Electrical Engineering and Computer Science, VSB – Technical University of Ostrava, Czech Republic. His major research interests are wireless communication in 5G, energy harvesting, performance of cognitive radio, physical layer security and NOMA.
\end{IEEEbiographynophoto}
\begin{IEEEbiographynophoto}
	{Symeon Chatzinotas} is currently Full Professor / Chief Scientist I and Head of the SIGCOM Research Group at SnT, University of Luxembourg. He is coordinating the research activities on communications and networking, acting as a PI for more than 20 projects and main representative for 3GPP, ETSI, DVB.
	
	In the past, he has been a Visiting Professor at the University of Parma, Italy, lecturing on “5G Wireless Networks”. He was involved in numerous R\&D projects for NCSR Demokritos, CERTH Hellas and CCSR, University of Surrey.
	
	He was the co-recipient of the 2014 IEEE Distinguished Contributions to Satellite Communications Award and Best Paper Awards at EURASIP JWCN, CROWNCOM, ICSSC. He has (co-)authored more than 450 technical papers in refereed international journals, conferences and scientific books.
	
	He is currently in the editorial board of the IEEE Transactions on Communications, IEEE Open Journal of Vehicular Technology and the International Journal of Satellite Communications and Networking.
\end{IEEEbiographynophoto}
\begin{IEEEbiographynophoto}
{Derrick Wing Kwan Ng (S'06-M'12-SM'17-F'21)} received the bachelor's degree with first-class honors and the Master of Philosophy (M.Phil.) degree in electronic engineering from the Hong Kong University of Science and Technology (HKUST) in 2006 and 2008, respectively. He received his Ph.D. degree from the University of British Columbia (UBC) in Nov. 2012. He was a senior postdoctoral fellow at the Institute for Digital Communications, Friedrich-Alexander-University Erlangen-N\"urnberg (FAU), Germany. He is now working as a Scientia Associate Professor at the University of New South Wales, Sydney, Australia.  His research interests include convex and non-convex optimization, physical layer security, IRS-assisted communication, UAV-assisted communication, wireless information and power transfer, and green (energy-efficient) wireless communications.
	
Dr. Ng has been listed as a Highly Cited Researcher by Clarivate Analytics (Web of Science) since 2018.  He received the Australian Research Council (ARC) Discovery Early Career Researcher Award 2017,  the IEEE Communications Society Stephen O. Rice Prize 2022,  the Best Paper Awards at the WCSP 2020,  2021,  IEEE TCGCC Best Journal Paper Award 2018, INISCOM 2018, IEEE International Conference on Communications (ICC) 2018, 2021,  IEEE International Conference on Computing, Networking and Communications (ICNC) 2016,  IEEE Wireless Communications and Networking Conference (WCNC) 2012, the IEEE Global Telecommunication Conference (Globecom) 2011, 2021 and the IEEE Third International Conference on Communications and Networking in China 2008.  He has been serving as an editorial assistant to the Editor-in-Chief of the IEEE Transactions on Communications from Jan. 2012 to Dec. 2019. He is now serving as an editor for the IEEE Transactions on Communications, the IEEE Transactions on Wireless Communications, and an area editor for the IEEE Open Journal of the Communications Society.
\end{IEEEbiographynophoto}
\begin{IEEEbiographynophoto}
	{H. Vincent Poor} (S'72, M'77, SM'82, F'87) received the Ph.D. degree in EECS from
		Princeton University in 1977. From 1977 until 1990, he was on the faculty of the
		University of Illinois at Urbana-Champaign. Since 1990 he has been on the faculty at
		Princeton, where he is currently the Michael Henry Strater University Professor. During
		2006 to 2016, he served as the dean of Princeton’s School of Engineering and Applied
		Science. He has also held visiting appointments at several other universities, including
		most recently at Berkeley and Cambridge. His research interests are in the areas of
		information theory, machine learning and network science, and their applications in
		wireless networks, energy systems and related fields. Among his publications in these
		areas is the forthcoming book Machine Learning and Wireless Communications.
		(Cambridge University Press). Dr. Poor is a member of the National Academy of
		Engineering and the National Academy of Sciences and is a foreign member of the
		Chinese Academy of Sciences, the Royal Society, and other national and international
		academies. He received the IEEE Alexander Graham Bell Medal in 2017.
\end{IEEEbiographynophoto}
\end{document}